\documentclass[10pt,conference]{IEEEtran}
\IEEEoverridecommandlockouts
\usepackage{cite,graphicx,amsmath,amssymb}
\usepackage{subfigure}
\usepackage{citesort}
\usepackage{fancyhdr}
\usepackage{dsfont}
\usepackage{array,color}
\usepackage{amsfonts}
\usepackage{hhline}
\usepackage{cases}
\newcommand{\dis}{\stackrel{d}{\sim}}

\newcommand{\eqla}{\stackrel{(a)}{=}}
\newcommand{\eqlb}{\stackrel{(b)}{=}}

\newtheorem{theorem}{Theorem}

\newtheorem{lemma}{Lemma}

\newtheorem{corollary}{Corollary}

\newtheorem{proposition}{Proposition}

\newcommand{\define}{\stackrel{\Delta}{=}}

\renewcommand{\baselinestretch}{1.0}

\begin{document}
%\title{Spatial Multiplexing with Limited Feedback in Ad Hoc Networks}
%\author{Yueping Wu${}^{\ddagger}$, Raymond H.\ Y.\ Louie${}^{H *}$,  Matthew R.\ McKay${}^{\ddagger}$, Iain B.\
%Collings$^{*}$\\
%${}^\ddagger${\small
% Department of Electronic and Computer
%Engineering, Hong Kong University of Science and Technology, Hong
%Kong}\\
%${}^H${\small
% Telecommunications Lab, School of Electrical and Information Engineering, University of Sydney, Australia} \\
%${}^*${\small
% Wireless Technologies Laboratory, ICT Centre, CSIRO,
%Sydney, Australia} }

%\title{Spatial Multiplexing with Limited Feedback in Ad Hoc Networks}
\title{User-Centric Interference Nulling in Downlink Multi-Antenna Heterogeneous Networks}
\author{
Yueping Wu$^{*}$, \, Ying Cui$^{\ddag}$, \,
Bruno Clerckx$^{*\dagger}$\\
\begin{minipage}{2\columnwidth}
\vspace{2mm}
\begin{center}
\small $^{*}$Department of Electrical and Electronic Engineering, Imperial College London\\
$^{\ddag}$Department of Electronic Engineering, Shanghai Jiao Tong University\\
$^{\dagger}$School of Electrical Engineering, Korea University\\
\end{center}
\end{minipage}
\thanks{Y. Cui was supported by the National Science Foundation of China grant 61401272. Y. Wu and B. Clerckx were partially supported by the Seventh Framework Programme
for Research of the European Commission under grant number HARP-318489.}
}

\maketitle

%\thispagestyle{fancyplain}
%\vspace*{-1.5cm}

 %%%%%%%%%%%%
\begin{abstract}
Heterogeneous networks (HetNets) have strong interference due to spectrum reuse. This affects the signal-to-interference ratio (SIR) of each user, and hence is one of the limiting factors of network performance.  However, in previous works, interference management approaches in HetNets are mainly based on interference level, and thus cannot effectively utilize the limited resource to improve network performance. In this paper, we propose a user-centric interference nulling (IN) scheme in downlink two-tier HetNets to improve network performance by improving each user's SIR. This scheme has three design parameters: the maximum degree of freedom for IN (i.e., maximum IN DoF), and the IN thresholds for the macro and pico users, respectively. Using tools from stochastic geometry, we first obtain a tractable expression of the coverage (equivalently outage) probability. Then, we characterize the asymptotic behavior of the outage probability in the high reliability regime. The asymptotic results show that the maximum IN DoF can affect the order gain of the asymptotic outage probability, while the IN thresholds only affect the coefficient of the asymptotic outage probability. Moreover, we show that the IN scheme can linearly improve the outage performance, and characterize the optimal maximum IN DoF which minimizes the asymptotic outage probability. 
%In addition, simulation results show that the IN scheme investigated in this paper can achieve a large coverage probability gain over the non-IN scheme.
\end{abstract}

\section{Introduction}
The modern wireless networks have seen a significant growth of high data rate applications. The conventional cellular solution, which comprises of high power base stations (BSs), cannot scale with the increasing data rate demand. One solution is the deployment of low power small cell BSs overlaid with conventional large power macro-BSs, so called heterogeneous networks (HetNets). HetNets are capable of aggressively reusing existing spectrum assets to support high data rate applications. However, spectrum reuse in HetNets causes strong interference. This affects the signal-to-interference ratio (SIR) of each user, and hence is one of the limiting factors of network performance. Interference management techniques are thus desirable in HetNets. One such technique is interference cooperation. For example, in \cite{nigam14}, a cooperation strategy among a fixed number of strongest BSs of each user is investigated for HetNets. Reference \cite{Tanbourgi14} proposes another interference cooperation strategy among BSs whose long-term average received powers (RP) are above a threshold. However, both cooperation strategies do not take into account the strength of each user's desired signal, and thus cannot effectively improve each user's SIR. Moreover, both \cite{nigam14} and  \cite{Tanbourgi14} only consider single-antenna BSs. 

%Another densifying approach is the deployment of multiple antennas at each BS. Besides increasing the network throughput, with multiple antennas, more effective interference management techniques can be implemented. For example, an interference coordination scheme was investigated in \cite{lee15} for multi-antenna cellular networks, where each user is served by its nearest BS and avoids interference from the next $n-1$ nearest BSs by utilizing zero-forcing beamforming (ZFBF) precoder at each BS. In \cite{huang13}, multi-antenna BSs in a coordination cluster which is formed by a hexagonal lattice mitigate interference to other users in the same cluster by using ZFBF precoder. A similar interference coordination scheme was investigated in \cite{akoum13}, but the cluster centers are formed by a homogeneous Poisson point process (PPP). Recently, a novel interference coordination strategy was proposed in \cite{li14IC}, where an interference nulling (IN) region is set for each user based on the power ratio of the desired signal and interference from each BS, and the multi-antenna BSs in the IN region of each user use ZFBF precoder to suppress interference to this user. Note that all these interference coordination schemes were investigated in single-tier cellular networks. 

Deploying multiple antennas at each BS in HetNets can further improve network performance. With multiple antennas, besides boosting the desired signal to each user, more effective interference management techniques can be implemented. Interference coordination strategies using multiple antennas have been extensively investigated in single-tier cellular networks (see \cite{huang13,li14IC} and the references therein). However, there are very limited works on the analysis of interference management techniques in multi-antenna HetNets. In \cite{Adhikary14}, the authors consider a HetNet with a single multi-antenna macro-BS and multiple small-BSs, where the multiple antennas at the macro-BS are used for serving its scheduled users as well as mitigating the interference to the receivers in the small cells using different interference coordination schemes. These schemes are analyzed and shown to improve the performance of the HetNet. However, since only one macro-BS is considered, the analytical results obtained in \cite{Adhikary14} cannot reflect the macro-tier interference, and thus cannot offer accurate insights for practical HetNets. In \cite{xia13}, a \emph{fixed} number of BSs form a cluster, and adopt an interference coordination scheme in downlink multi-antenna HetNets, where the BSs mitigate interference to users in the cluster.  Bounds of the coverage probability are derived based on an unrealistic assumption that the BSs in each cluster are the strongest BSs of all the users in this cluster. This scheme cannot effectively improve each user's SIR, as it does not consider each user's desired signal strength.  Recently, a novel user-centric IN scheme, which takes account of each user's desired signal strength  and interference level, is proposed and analyzed for single-tier cellular networks with multi-antenna BSs in \cite{li14IC}. However, in \cite{li14IC}, the maximum degree of freedom for IN (IN DoF) at each BS is not adjustable. This scheme thus cannot properly exploit the limited resource in single-tier cellular networks. Moreover, directly applying the approach in \cite{li14IC} to HetNets cannot fully exploit different properties of the macro and pico users in HetNets. 

In this paper, we consider the downlink two-tier multi-antenna HetNets and propose a user-centric IN scheme to improve network performance by improving each user's SIR. In this scheme, each  macro (pico) user first sends an IN request to a macro-BS\footnote{Note that, compared to pico-BSs, macro-BSs usually cause stronger interference due to larger transmit power, and each macro-BS normally has a larger number of transmit antennas (i.e., a better capability of performing spatial cancellation). Thus, it is more advisable to perform IN at macro-BSs.} if the power ratio of its desired signal and the interference from the macro-BS, referred to as the signal-to-individual-interference ratio (SIIR), is below an IN threshold for macro (pico) users. Then, each macro-BS utilizes zero-forcing beamforming (ZFBF) precoder to avoid interference to at most $U$ users which send IN requests to it as well as boost the desired signal to its scheduled user. This scheme has three design parameters: the maximum IN DoF $U$, and the IN thresholds for the macro and pico users, respectively. In general, the investigation of interference management techniques in multi-antenna HetNets is very challenging, mainly due to i) the statistical dependence among macro-BSs and pico-BSs \cite{Adhikary14}, ii) the complex distribution of desired signal using multi-antenna communication schemes, and iii) the complicated interference distribution caused by interference management techniques (e.g., beamforming). In this paper, by adopting appropriate approximations and utilizing tools from stochastic geometry, we first obtain a tractable expression of the coverage (equivalently outage) probability. Then, we characterize the asymptotic behavior of the outage probability in the high reliability regime. The asymptotic results show that the maximum IN DoF can affect the order gain of the asymptotic outage probability, while the IN thresholds only affect the coefficient of the asymptotic outage probability. Moreover, we show that the IN scheme can linearly improve the outage performance, and characterize the optimal maximum IN DoF which minimizes the asymptotic outage probability. The analytical results obtained in this paper provide valuable design insights for practical HetNets.

\section{System Model}
\subsection{Downlink Two-Tier Heterogeneous Networks}
We consider a two-tier HetNet where a macro-cell tier is overlaid with a pico-cell tier, as shown in Fig.\ \ref{fig:model}. The locations of the macro-BSs and the pico-BSs are spatially distributed as two independent homogeneous Poisson point processes (PPPs) $\Phi_{1}$ and $\Phi_{2}$ with densities $\lambda_{1}$ and $\lambda_{2}$, respectively. The locations of the users are also distributed as an independent homogeneous PPP $\Phi_{u}$ with density $\lambda_{u}$. Without loss of generality, denote the macro-cell tier as the $1$st tier and the pico-cell tier as the $2$nd tier. We focus on the downlink scenario. Each macro-BS has $N_{1}$ antennas with total transmission power $P_{1}$, each pico-BS has $N_{2}$ antennas with total transmission power $P_{2}$, and each user has a single antenna. We assume $N_{1}>N_{2}$. We consider both large-scale fading and small-scale fading. Specifically, due to large-scale fading, transmitted signals from the $j$th tier with distance $r$ are attenuated by a factor $\frac{1}{r^{\alpha_{j}}}$, where $\alpha_{j}>2$ is the path loss exponent of the $j$th tier and $j=1,2$. For small-scale fading, we assume Rayleigh fading channels.

\subsection{User Association}\label{subset:user_asso}
We assume open access \cite{nigam14}. User $i$ (denoted as $u_{i}$) is associated with the BS which provides the maximum \emph{long-term average} RP among all the macro-BSs and pico-BSs. This associated BS is called the \emph{serving BS} of user $i$. Note that within each tier, the nearest BS to user $i$ provides the strongest long-term average RP in this tier. User $i$ is thus associated with (the nearest BS in) the $j^{*}_{i}$th tier, if\footnote{In the user association procedure, the first antenna is normally used to transmit signal (using the total transmission power of each BS) for RP determination according to LTE standards.} $j_{i}^{*}=  {\arg\:\max}_{j\in\{1,2\}}P_{j}Z_{i,j}^{-\alpha_{j}}$,
where $Z_{i,j}$ is the distance between user $i$ and its nearest BS in the $j$th tier. We refer to the users associated with the macro-cell tier as the \emph{macro-users}, denoted as $\mathcal{U}_{1}=\left\{u_{i}|P_{1}Z_{i,1}^{-\alpha_{1}}\ge P_{2}Z_{i,2}^{-\alpha_{2}}\right\}$, and the users associated with the pico-cell tier as the \emph{pico-users}, denoted as $\mathcal{U}_{2}=\left\{u_{i}|P_{2}Z_{i,2}^{-\alpha_{2}}>P_{1}Z_{i,1}^{-\alpha_{1}}\right\}$. All the users can be partitioned into two disjoint user sets: $\mathcal{U}_{1}$ and $\mathcal{U}_{2}$. After user association, each BS schedules its associated users according to TDMA, i.e., scheduling one user in each time slot. Hence, there is no intra-cell interference. 

%According to the above mentioned user association policy, all the users can be partitioned into the following two disjoint user sets:
%%1) \emph{the set of macro-users} $\mathcal{U}_{1}=\left\{u_{i}|P_{1}Z_{i,1}^{-\alpha_{1}}\ge BP_{2}Z_{i,2}^{-\alpha_{2}}\right\}$, 2) \emph{the set of unoffloaded pico-users} $\mathcal{U}_{2\bar{O}}=\left\{u_{i}|P_{2}Z_{i,2}^{-\alpha_{2}}>P_{1}Z_{i,1}^{-\alpha_{1}}\right\}$, 3) \emph{the set of offloaded users} $\mathcal{U}_{2O}=\left\{u_{i}|P_{2}Z_{i,2}^{-\alpha_{2}}\le P_{1}Z_{i,1}^{-\alpha_{1}}<BP_{2} Z_{i,2}^{-\alpha_{2}}\right\}$,
%\begin{enumerate}
%\item the set of \emph{macro-users}:\\ $\mathcal{U}_{1}=\left\{u_{i}|P_{1}Z_{i,1}^{-\alpha_{1}}\ge BP_{2}Z_{i,2}^{-\alpha_{2}}\right\}$\;,
%\item the set of \emph{ pico-users}:\\ $\mathcal{U}_{2}=\left\{u_{i}|P_{2}Z_{i,2}^{-\alpha_{2}}>P_{1}Z_{i,1}^{-\alpha_{1}}\right\}$\;.
%%\item \emph{offloaded user without IN} set: \\$\mathcal{U}_{2O\bar{C}}=\left\{u_{i}|P_{2}Z_{i,2}^{-\alpha_{k}}\le P_{1}Z_{i,1}^{-\alpha_{1}}<BP_{2} Z_{i,2}^{-\alpha_{2}}\;{\rm and}\;u_{i}\;{\rm is\;not\; nulled\; interference}\right\}$\;.
%\end{enumerate}
%where the macro-users are associated with the maco-cell tier, and the pico-users are associated with the pico-cell tier.

\subsection{Performance Metric}
In this paper, we study the performance of the typical user denoted as $u_{0}$, which is a scheduled user located at the origin \cite{Tanbourgi14}. Since HetNets are interference-limited, we ignore the thermal noise in the analysis of this paper. Note that the analytical results with thermal noise can be calculated in a similar way. We investigate the \emph{coverage probability} of $u_{0}$, which is defined as the probability that the SIR of $u_{0}$ is larger than a threshold \cite{nigam14}, and can be mathematically written as
\begin{align}\label{eq:CP_def}
\mathcal{S}(\beta)&\define{\rm Pr}\left({\rm SIR}_{0}>\beta\right)
\end{align} where $\beta$ is the SIR threshold.
%Based on (\ref{eq:CP_def}), the \emph{outage probability} of $u_{0}$, which is defined as the probability that the SIR of $u_{0}$ is smaller than a threshold, can be written as $1-\mathcal{S}\left(\beta\right)$.

\section{Interference Nulling Scheme}
%In HetNets, each user experiences strong interference, especially the interference from macro-BSs with larger transmit power. This may cause small SIR to the user, especially to  users which have poor desired signal.
In this section, we first elaborate on a user-centric IN scheme to avoid interference from some macro-BSs which generate dominant interference. Then, we obtain some distributions related to this scheme.

\begin{figure}[t] \centering
\includegraphics[width=1\columnwidth]{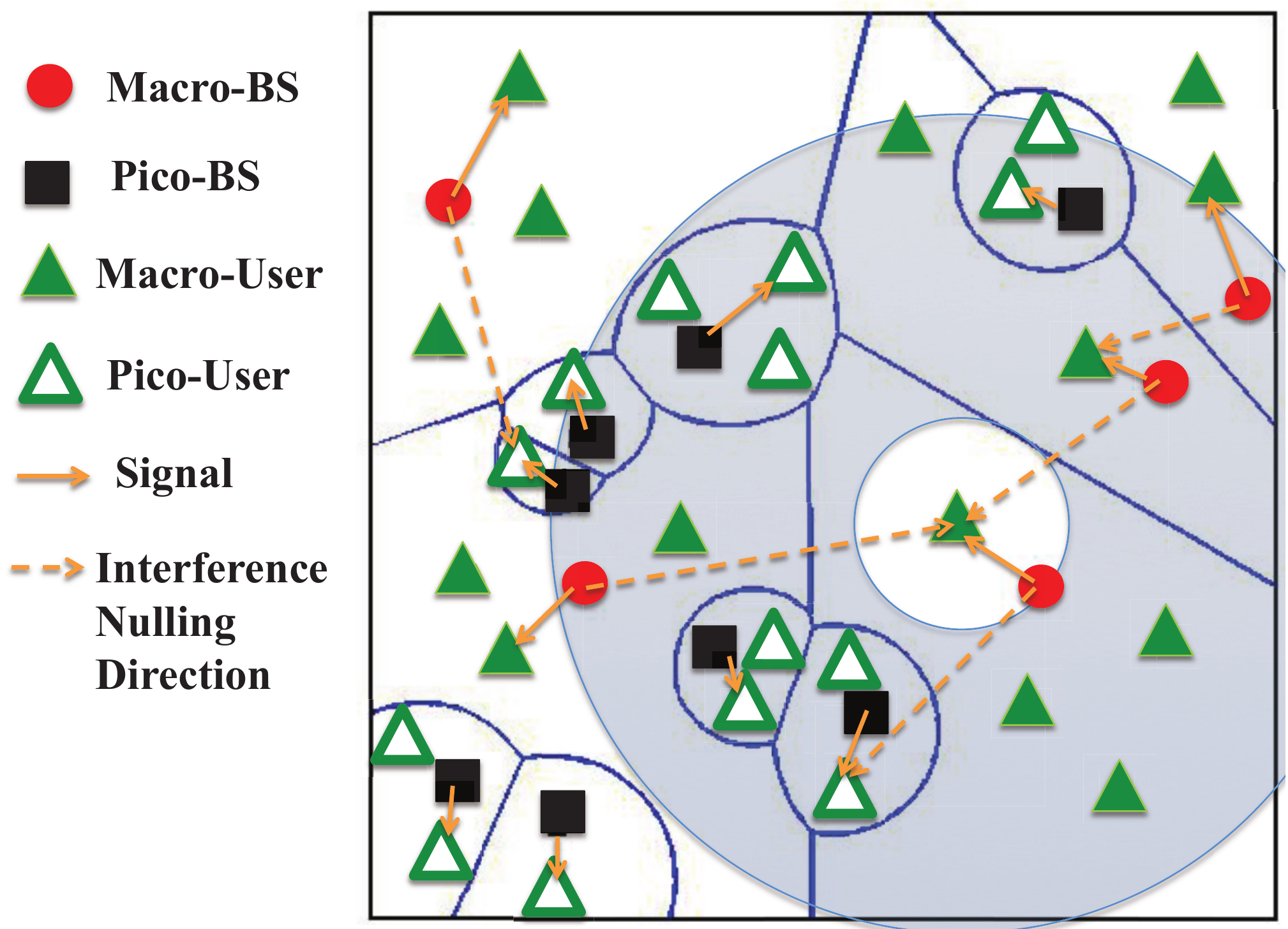}
\caption{System Model ($U=1$)}\label{fig:model}
\vspace{-1mm}
\end{figure}

\subsection{IN Scheme Description}\label{subsec:IN_descp}
First, we refer to a macro-BS which causes the SIIR at a scheduled user $i$ in the $j$th tier ($u_{i}\in\mathcal{U}_{j}$) below threshold $T_{j}>1$ as a \emph{potential IN macro-BS} of $u_{i}$, where $j=1,2$. We refer to $T_{j}$ as the IN threshold for the $j$th tier. Mathematically, macro-BS $\ell$ is a potential IN macro-BS of scheduled user $u_{i}\in\mathcal{U}_{j}$ if $\frac{P_{j}Z_{i,j}^{-\alpha_{j}}}{P_{1}|D_{1,\ell i}|^{-\alpha_{1}}}<T_{j}$, where $|D_{1,\ell i}|$ is the distance from macro-BS $\ell$ to $u_{i}$. Note that $T_{1}$ and $T_{2}$ are two design parameters of the IN scheme. In each time slot, each scheduled user sends IN requests to all of its potential IN macro-BSs. We refer to the scheduled users which send IN requests to macro-BS $\ell$ as the \emph{potential IN users} of macro-BS $\ell$ (in this time slot). Consider a particular time slot. Let $K_{\ell}$ denote the number of the potential IN users of macro-BS $\ell$. Consider two cases in the following. i) If $K_{\ell}>0$, macro-BS $\ell$ makes use of at most $U$ ($0<U<N_{1}$) DoF to perform IN to some of its potential IN users. Note that $U$ is another design parameter of this IN scheme. In particular, if $0<K_{\ell}\le U$, macro-BS $\ell$ can perform IN to all of its $K_{\ell}$ potential IN users using $K_{\ell}$ DoF; if $K_{\ell}>U$, macro-BS $\ell$ randomly selects $U$ out of its $K_{\ell}$ potential IN users according to the uniform distribution, and perform IN to the selected $U$ users using $U$ DoF. Hence, if $K_{\ell}>0$, macro-BS $\ell$ performs IN to $u_{{\rm IN},\ell}\define\min\left(U,K_{\ell}\right)$ potential IN users (referred to as the \emph{IN users} of macro-BS $\ell$) using $u_{{\rm IN},\ell}$ DoF (referred to as the \emph{IN DoF}). ii) If $K_{\ell}=0$, macro-BS $\ell$ does not perform IN. In this case, we let $u_{{\rm IN},\ell}=0$. In both cases, $N_{1}-u_{{\rm IN},\ell}$ DoF at macro-BS $\ell$ is used for boosting the desired signal to its scheduled user.

Now, we introduce the precoding vectors at macro-BSs in the IN scheme. Consider two cases in the following. i) If $K_{\ell}>0$, macro-BS $\ell$ utilizes the low-complexity ZFBF precoder to serve its scheduled user and simultaneously perform IN to its $u_{{\rm IN},0}$ IN users. Specifically,  denote $\mathbf{H}_{1,\ell}=\left[\mathbf{h}_{1,\ell}\; \mathbf{g}_{1,\ell1}\;\ldots\;\mathbf{g}_{1,\ell u_{{\rm IN},\ell}}\right]^{\dagger}$, where\footnote{The notation $X \dis Y$ means that $X$ \emph{is distributed as} $Y$.} $\mathbf{h}_{1,\ell}\dis\mathcal{CN}_{N_{1},1}\left(\mathbf{0}_{N_{1}\times 1},\mathbf{I}_{N_{1}}\right)$ denotes the channel vector between macro-BS $\ell$ and its scheduled user, and $\mathbf{g}_{1,\ell i}\dis \mathcal{CN}_{N_{1},1}\left(\mathbf{0}_{N_{1}\times 1},\mathbf{I}_{N_{1}}\right)$ denotes the channel vector between macro-BS $\ell$ and its $i$th IN user $(i=1,\ldots,u_{{\rm IN},\ell})$. The ZFBF precoding matrix at macro-BS $\ell$ is designed to be $\mathbf{W}_{1,\ell}=\mathbf{H}_{1,\ell}^{\dagger}\left(\mathbf{H}_{1,\ell}\mathbf{H}_{1,\ell}^{\dagger}\right)^{-1}$ and the ZFBF vector at macro-BS $\ell$ is designed to be $\mathbf{f}_{1,\ell}=\frac{\mathbf{w}_{1,\ell}}{\|\mathbf{w}_{1,\ell}\|}$, where $\mathbf{w}_{1,\ell}$ is the first column of $\mathbf{W}_{1,\ell}$. ii) If $K_{\ell}=0$, macro-BS $\ell$ uses the maximal ratio transmission (MRT) precoder to serve its scheduled user, which is a special case of the ZFBF precoder introduced for $K_{\ell}>0$ and can be readily obtained from it by letting $u_{{\rm IN},\ell}=0$, i.e., $\mathbf{H}_{1,\ell}=\mathbf{h}_{1,\ell}^{\dagger}$. Next, we introduce the precoding vectors at pico-BSs. Each pico-BS utilizes the MRT precoder to serve its scheduled user. Specifically, the beamforming vector at pico-BS $\ell$ is $\mathbf{f}_{2,\ell}=\frac{\mathbf{h}_{2,\ell}}{\left\|\mathbf{h}_{2,\ell}\right\|}$, 
%\begin{align}\label{eq:BFvec_pico}
%\mathbf{f}_{2,i}=\frac{\mathbf{h}_{2,i}}{\left\|\mathbf{h}_{2,i}\right\|}
%\end{align} 
where $\mathbf{h}_{2,\ell}\dis\mathcal{CN}_{N_{2},1}\left(\mathbf{0}_{N_{2}\times 1},\mathbf{I}_{N_{2}}\right)$ denotes the channel vector between pico-BS $\ell$ and its scheduled user. Note that the non-IN scheme can be included in the IN scheme as a special case by letting $T_{1}=T_{2}=1$ and $U=0$.\footnote{All the analytical results in this paper hold for $T_{1}=T_{2}=1$ and $U=0$.} 

We now obtain the SIR of the typical user. Under the above IN scheme, $u_{0}$ experiences three types of interference: 1) residual aggregated interference $I_{j,1C}$ from its potential IN macro-BSs $\Phi_{1C}$ which do not select $u_{0}$ for IN, 2) aggregated interference $I_{j,1O}$ from interfering macro-BSs $\Phi_{1O}$ which are not its potential IN macro-BSs, and 3) aggregated interference $I_{j,2}$ from all interfering pico-BSs.  Specifically, the SIR of the typical user $u_{0}\in\mathcal{U}_{j}$ $(j=1,2)$ is given by\footnote{The index of the typical user and its serving BS is $0$.}
\begin{align}\label{eq:SIRj0}
{\rm SIR}_{j,0}&=\frac{\frac{P_{j}}{Y_{j}^{\alpha_{j}}}\left|\mathbf{h}_{j,00}^{\dagger}\mathbf{f}_{j,0}\right|^{2}}{P_{1}I_{j,1C}+P_{1}I_{j,1O}+P_{2}I_{j,2}}
\end{align}
where  $\mathbf{h}_{j,00}$ is the channel vector between $u_{0}\in\mathcal{U}_{j}$ and its serving BS $B_{j,0}$, $Y_{j}$ is the distance between $u_{0}$ and $B_{j,0}$, $\mathbf{f}_{j,0}$ is the beamforming vector at $B_{j,0}$, with $\left|\mathbf{h}_{j,00}^{\dagger}\mathbf{f}_{j,0}\right|^{2}\dis{\rm Gamma}\left(M_{j},1\right)$, $M_{1}=N_{1}-u_{{\rm IN},0}$ and $M_{2}=N_{2}$. Here, $I_{j,1C}=\sum_{\ell\in\Phi_{1C}}\left|D_{1,\ell0}\right|^{-\alpha_{1}}\left|\mathbf{h}_{1,\ell0}^{\dagger}\mathbf{f}_{1,\ell}\right|^{2}$, $I_{j,1O}=\sum_{\ell\in\Phi_{1O}}\left|D_{1,\ell0}\right|^{-\alpha_{1}}\left|\mathbf{h}_{1,\ell0}^{\dagger}\mathbf{f}_{1,\ell}\right|^{2}$, and $I_{j,2}=\sum_{\ell\in\Phi_{2}\backslash B_{j,0}}|D_{2,\ell0}|^{-\alpha_{2}}\left|\mathbf{h}_{2,\ell0}^{\dagger}\mathbf{f}_{2,\ell}\right|^{2}$, where $\mathbf{h}_{j,\ell0}$ is the channel vector between BS $\ell$ in the $j$th tier and $u_{0}$, $\mathbf{f}_{j,\ell}$ is the beamforming vector at BS $\ell$ in the $j$th tier, and $\left|\mathbf{h}_{j,\ell0}^{\dagger}\mathbf{f}_{j,\ell}\right|^{2}\dis{\rm Gamma}(1,1)$.

\subsection{Preliminary Results}\label{subsec:prelim}
In this part, we evaluate some distributions related to the IN scheme, which will be used to calculate the coverage probability in (\ref{eq:CP_def}). We first calculate the probability mass function (p.m.f.) of the number of the potential IN users of $u_{0}$'s serving macro-BS (when $u_{0}\in\mathcal{U}_{1}$), denoted as $K_{0}$. The p.m.f. of $K_{0}$ depends on the point processes formed by the scheduled macro and pico users, which are related to but not PPPs \cite{bai13}. For analytical tractability, we approximate the scheduled macro and pico users as two independent PPPs with densities $\lambda_{1}$ and $\lambda_{2}$, respectively.\footnote{Note that approximating the scheduled users as a homogeneous PPP has been considered in existing papers (see e.g., \cite{bai13}). Moreover, simulation results in Section \ref{subsec:CP_general} will verify the accuracy of this approximation.} Then, we have the p.m.f. of $K_{0}$ as follows:
\begin{lemma}\label{lem:num_req_pmf}
The p.m.f. of $K_{0}$ is given by
\begin{align}\label{eq:K_pmf}
{\rm Pr}\left(K_{0}=k\right)\approx\frac{\bar{L}^{k}}{k!}\exp\left(-\bar{L}\right)
\end{align}
where $\bar{L}=\bar{L}_{1}+\bar{L}_{2}$ with
\begin{align}
\bar{L}_{j}=&2\pi\lambda_{j}\int_{0}^{\infty}r\int_{\left(\frac{P_{j}}{P_{1}T_{j}}\right)^{\frac{1}{\alpha_{j}}}r^{\frac{\alpha_{1}}{\alpha_{j}}}}^{\left(\frac{P_{j}}{P_{1}}\right)^{\frac{1}{\alpha_{j}}}r^{\frac{\alpha_{1}}{\alpha_{j}}}}f_{Y_{j}}(y){\rm d}y{\rm d}r\;,\;j=1,2.
\end{align}
Here, $f_{Y_{j}}(y)$ ($j=1,2$) are given as follows:
\begin{align}
f_{Y_{1}}(y)&=\frac{2\pi\lambda_{1}}{\mathcal{A}_{1}}y\exp\left(-\pi\left(\lambda_{1}y^{2}+\lambda_{2}\left(\frac{P_{2}}{P_{1}}\right)^{\frac{2}{\alpha_{2}}}y^{\frac{2\alpha_{1}}{\alpha_{2}}}\right)\right)\label{eq:pdfY1}\\
f_{Y_{2}}(y)&=\frac{2\pi\lambda_{2}}{\mathcal{A}_{2}}y\exp\left(-\pi\left(\lambda_{1}\left(\frac{P_{1}}{P_{2}}\right)^{\frac{2}{\alpha_{1}}}y^{\frac{2\alpha_{2}}{\alpha_{1}}}+\lambda_{2}y^{2}\right)\right)\label{eq:pdfY2}
\end{align}
where $\mathcal{A}_{j}\define {\rm Pr}\left(u_{0}\in\mathcal{U}_{j}\right)$ ($j=1,2$) are given in (\ref{eq:A1}) and (\ref{eq:A2}) at the top of the next page.
\begin{figure*}[!t]
\begin{align}
&\mathcal{A}_{1}=2\pi\lambda_{1}\int_{0}^{\infty}z\exp\left(-\pi\lambda_{1}z^{2}\right)\exp\left(-\pi\lambda_{2}\left(\frac{P_{2}}{P_{1}}\right)^{\frac{2}{\alpha_{2}}}z^{\frac{2\alpha_{1}}{\alpha_{2}}}\right){\rm d}z\label{eq:A1}\\
&\mathcal{A}_{2}=2\pi\lambda_{2}\int_{0}^{\infty}z\exp\left(-\pi\lambda_{2}z^{2}\right)\exp\left(-\pi\lambda_{1}\left(\frac{P_{1}}{P_{2}}\right)^{\frac{2}{\alpha_{1}}}z^{\frac{2\alpha_{2}}{\alpha_{1}}}\right){\rm d}z\label{eq:A2}
\end{align}
\normalsize \hrulefill
\end{figure*}
\end{lemma}
\begin{proof}
See Appendix \ref{proof:pmf_Lbar}.
\end{proof}

%Note that $\bar{L}$ denotes the average number of IN requests received by the serving BS of $u_{0}$, and $\bar{L}_{j}$ ($j=1,2$) denotes the average number of IN requests received by the serving BS of $u_{0}$ from the scheduled users in the $j$th tier.

%From (\ref{eq:Kbar_uneql}), we can easily see that $\bar{K}_{j}$ increases with $T_{j}$, which can be explained by noting that as $T_{j}$ increases, each macro-BS can be in the IN region of more scheduled users in the $j$th tier. Moreover, since $\lambda_{1}<\lambda_{2}$, we can show that $\left(1+\frac{\lambda_{2}}{\lambda_{1}}\left(\frac{P_{2}}{P_{1}}\right)^{\frac{2}{\alpha}}\right)^{-1}<\left(\frac{\lambda_{1}}{\lambda_{2}}+\left(\frac{P_{2}}{P_{1}}\right)^{\frac{2}{\alpha}}\right)^{-1}$. Hence, when $\alpha_{1}=\alpha_{2}=\alpha$, from (\ref{eq:Kbar_eql}), we know that the increasing rate of $\bar{K}_{1}$ w.r.t.\ $T_{1}$ is smaller than that of $\bar{K}_{2}$ w.r.t.\ $T_{2}$.

Next, we calculate the p.m.f. of $u_{{\rm IN},0}\define\min\left(U,K_{0}\right)$ based on \emph{Lemma \ref{lem:num_req_pmf}}, which is given as follows:
\begin{lemma}\label{lem:pmf_uIN0}
The p.m.f. of $u_{{\rm IN},0}$ is given by
\begin{align}
{\rm Pr}\left(u_{{\rm IN},0}=u\right)=
\begin{cases}
&{\rm Pr}\left(K_{0}=u\right)\;,\hspace{12mm} {\rm for}\;0\le u<U \\
&\sum_{k=U}^{\infty}{\rm Pr}\left(K_{0}=k\right)\;,\hspace{2mm}{\rm for}\;u=U
\end{cases}
\;.\notag
\end{align}
\end{lemma}

Let $p_{c}\left(U,T_{1},T_{2}\right)$ denote the probability that a randomly selected (according to the uniform distribution) potential IN macro-BS of $u_{0}$ selects $u_{0}$ for IN. Note that the event that $u_{0}$ sends IN requests and the event that all the other scheduled users send IN requests are dependent. For analytical tractability, we approximate these two events as independent events. Then, we have $p_{c}\left(U,T_{1},T_{2}\right)\approx {\rm E}\left[\min\left\{\frac{U}{K_{0}+1},1\right\}\right]$, which can be calculated as follows:
\begin{lemma}\label{lem:INprob}
The IN probability is given by
\begin{align}
p_{c}\left(U,T_{1},T_{2}\right)
%\approx{\rm E}\left[\frac{1}{K_{0}+1}\right]\notag\\
\approx&\exp\left(-\bar{L}\right)\left(\sum_{k=0}^{U-1}\frac{\bar{L}^{k}}{k!}+U\sum_{k=U}^{\infty}\frac{\bar{L}^{k}}{(k+1)!}\right)\;.\notag
\end{align}
\end{lemma}
\begin{proof}
See Appendix \ref{proof:lem_INprob}
\end{proof}

Note that different potential IN macro-BSs of $u_{0}$ selects $u_{0}$ for IN dependently (as the numbers of the potential IN users of these macro-BSs depend on the locations of these macro-BSs and are thus dependent). For analytical tractability, we assume that different potential IN macro-BSs of $u_{0}$ select $u_{0}$ for IN independently. Using \emph{independent thinning}, $\Phi_{1C}$ can be \emph{approximated} by a homogeneous PPP with density $p_{\bar{c}}\left(U,T_{1},T_{2}\right)\lambda_{1}$, where $p_{\bar{c}}\left(U,T_{1},T_{2}\right)\define1-p_{c}\left(U,T_{1},T_{2}\right)$.

\section{Coverage Probability Analysis}
In this section, we investigate the coverage probability in the general and small SIR threshold regimes, respectively.

\subsection{General SIR Threshold Regime}\label{subsec:CP_general}
As discussed in Section \ref{subset:user_asso}, the typical user $u_{0}$ is in one of two disjoint user sets: $\mathcal{U}_{1}$ and $\mathcal{U}_{2}$. By utilizing tools from stochastic geometry and the preliminary results obtained in Section \ref{subsec:prelim}, we have the following theorem:
%\footnote{Due to page limit, we omit all the proofs. Please refer to \cite{wuISIT15} for details.}

\begin{figure*}[!t]
\begin{align}
&\mathcal{S}_{1}(\beta,U,T_{1},T_{2})=\sum_{u=0}^{U}{\rm Pr}\left(u_{{\rm IN},0}=u\right)\int_{0}^{\infty}\sum_{n=0}^{N_{1}-u-1}\frac{1}{n!}\sum_{\left(n_{a}\right)_{a=1}^{3}\in\mathcal{N}_{n}}\binom{n}{n_{1},n_{2},n_{3}}\mathcal{\tilde{L}}^{(n_{1})}_{I_{1,1C}}\left(\beta y^{\alpha_{1}},y,T^{\frac{1}{\alpha_{1}}}y\right)\mathcal{\tilde{L}}^{(n_{2})}_{I_{1,1O}}\left(\beta y^{\alpha_{1}},T_{1}^{\frac{1}{\alpha_{1}}}y\right)\notag\\
&\hspace{6.6cm}\times\mathcal{\tilde{L}}^{(n_{3})}_{I_{1,2}}\left(\beta\frac{P_{2}}{P_{1}}y^{\alpha_{1}},\left(\frac{P_{2}}{P_{1}}\right)^{\frac{1}{\alpha_{2}}}y^{\frac{\alpha_{1}}{\alpha_{2}}}\right)f_{Y_{1}}(y){\rm d}y\label{eq:CP1}\\
&\mathcal{S}_{2}(\beta,U,T_{1},T_{2})=\int_{0}^{\infty}\sum_{n=0}^{N_{2}-1}\frac{1}{n!}\sum_{(n_{a})_{a=1}^{3}\in\mathcal{N}_{n}}\binom{n}{n_{1},n_{2},n_{3}}\mathcal{\tilde{L}}^{(n_{1})}_{I_{2,1C}}\left(\beta\frac{P_{1}}{P_{2}}y^{\alpha_{2}},\left(\frac{P_{1}}{P_{2}}\right)^{\frac{1}{\alpha_{1}}}y^{\frac{\alpha_{2}}{\alpha_{1}}},\left(\frac{P_{1}T_{2}}{P_{2}}\right)^{\frac{1}{\alpha_{1}}}y^{\frac{\alpha_{2}}{\alpha_{1}}}\right)\notag\\
&\hspace{6cm}\times\mathcal{\tilde{L}}^{(n_{2})}_{I_{2,1O}}\left(\beta\frac{P_{1}}{P_{2}}y^{\alpha_{2}},\left(\frac{P_{1}T_{2}}{P_{2}}\right)^{\frac{1}{\alpha_{1}}}y^{\frac{\alpha_{2}}{\alpha_{1}}}\right)\mathcal{\tilde{L}}^{(n_{3})}_{I_{2,2}}\left(\beta y^{\alpha_{2}},y\right)f_{Y_{2}}(y){\rm d}y\label{eq:CP2}\\
&\tilde{\mathcal{L}}_{I_{j,1C}}^{(n)}(s,r_{j,1C},r_{j,1O})=\mathcal{L}_{I_{j,1C}}(s,r_{j,1C},r_{j,1O})\sum_{(m_{a})_{a=1}^{n}\in\mathcal{M}_{n}}\frac{n!}{\prod_{a=1}^{n}m_{a}!}\notag\\
&\hspace{0mm}\times\prod_{a=1}^{n}\left(\frac{2\pi}{\alpha_1}p_{\bar{c}}\left(U,T_{1},T_{2}\right)\lambda_{1}s^{\frac{2}{\alpha_{1}}}\left(B^{'}\left(1+\frac{2}{\alpha_{1}},a-\frac{2}{\alpha_{1}},\frac{1}{1+sr_{j,1C}^{-\alpha_{1}}}\right)-B^{'}\left(1+\frac{2}{\alpha_{1}},a-\frac{2}{\alpha_{1}},\frac{1}{1+sr_{j,1O}^{-\alpha_{1}}}\right)\right)\right)^{m_{a}}\label{eq:LTdiff_I1in}\\
&\mathcal{\tilde{L}}^{(n)}_{I_{j,k}}(s,r_{j,k})=\mathcal{L}_{I_{j,k}}(s,r_{j,k})\sum_{(m_{a})_{a=1}^{n}\in\mathcal{M}_{n}}\frac{n!}{\prod_{a=1}^{n}m_{a}!}\prod_{a=1}^{n}\left(\frac{2\pi\lambda_{J(k)}}{\alpha_{J(k)}}s^{\frac{2}{\alpha_{J(k)}}}B^{'}\left(1+\frac{2}{\alpha_{J(k)}},a-\frac{2}{\alpha_{J(k)}},\frac{1}{1+\frac{s}{r_{j,k}^{\alpha_{J(k)}}}}\right)\right)^{m_{a}}\label{eq:LTdiff_1O2}\\
&\mathcal{L}_{I_{j,1C}}(s,r_{j,1C},r_{j,1O})=\exp\left(-\frac{2\pi}{\alpha_{1}}p_{\bar{c}}\left(U,T_{1},T_{2}\right)\lambda_{1}s^{\frac{2}{\alpha_{1}}}\left(B^{'}\left(\frac{2}{\alpha_{1}},1-\frac{2}{\alpha_{1}},\frac{1}{1+sr_{j,1C}^{-\alpha_{1}}}\right)-B^{'}\left(\frac{2}{\alpha_{1}},1-\frac{2}{\alpha_{1}},\frac{1}{1+sr_{j,1O}^{-\alpha_{1}}}\right)\right)\right)\label{eq:LT_1in}\\
&\mathcal{L}_{I_{j,k}}(s,r_{j,k})=\exp\left(-\frac{2\pi\lambda_{J(k)}}{\alpha_{J(k)}}s^{\frac{2}{\alpha_{J(k)}}}B^{'}\left(\frac{2}{\alpha_{J(k)}},1-\frac{2}{\alpha_{J(k)}},\frac{1}{1+\frac{s}{r_{j,k}^{\alpha_{J(k)}}}}\right)\right)\label{eq:LT_1O2}
%&\mathcal{A}_{1}=2\pi\lambda_{1}\int_{0}^{\infty}z\exp\left(-\pi\lambda_{1}z^{2}\right)\exp\left(-\pi\lambda_{2}\left(\frac{P_{2}}{P_{1}}\right)^{\frac{2}{\alpha_{2}}}z^{\frac{2\alpha_{1}}{\alpha_{2}}}\right){\rm d}z\label{eq:A1}\\
%&\mathcal{A}_{2}=2\pi\lambda_{2}\int_{0}^{\infty}z\exp\left(-\pi\lambda_{2}z^{2}\right)\exp\left(-\pi\lambda_{1}\left(\frac{P_{1}}{P_{2}}\right)^{\frac{2}{\alpha_{1}}}z^{\frac{2\alpha_{2}}{\alpha_{1}}}\right){\rm d}z\label{eq:A2}
\end{align}
\normalsize \hrulefill
\end{figure*}

\begin{theorem}[Coverage Probability]\label{thm:overall_CP}
Under design parameters $U$, $T_{1}$ and $T_{2}$, we have
%$1)$ coverage probability of a macro-user: $\mathcal{S}_{1}\left(\beta,U,T_{1},T_{2}\right)\define{\rm Pr}\left({\rm SIR}_{0}>\beta|u_{0}\in\mathcal{U}_{1}\right)$, given in Appendix,\\
%$2)$ coverage probability of a pico-user: $\mathcal{S}_{2}\left(\beta,U,T_{1},T_{2}\right)\define{\rm Pr}\left({\rm SIR}_{0}>\beta|u_{0}\in\mathcal{U}_{2}\right)$, given in Appendix,\\
%$3)$ overall coverage probability $\mathcal{S}\left(\beta,U,T_{1},T_{2}\right)=\mathcal{A}_{1}\mathcal{S}_{1}(\beta,U,T_{1},T_{2})+\mathcal{A}_{2}\mathcal{S}_{2}(\beta,U,T_{1},T_{2})$. \\
\begin{enumerate}
\item coverage probability of a macro-user: $\mathcal{S}_{1}\left(\beta,U,T_{1},T_{2}\right)\define{\rm Pr}\left({\rm SIR}_{0}>\beta|u_{0}\in\mathcal{U}_{1}\right)$, given in (\ref{eq:CP1}) at the top of the next page,
\item coverage probability of a pico-user: $\mathcal{S}_{2}\left(\beta,U,T_{1},T_{2}\right)\define{\rm Pr}\left({\rm SIR}_{0}>\beta|u_{0}\in\mathcal{U}_{2}\right)$, given in (\ref{eq:CP2}) at the top of the next page,
\item overall coverage probability $\mathcal{S}\left(\beta,U,T_{1},T_{2}\right)=\mathcal{A}_{1}\mathcal{S}_{1}(\beta,U,T_{1},T_{2})+\mathcal{A}_{2}\mathcal{S}_{2}(\beta,U,T_{1},T_{2})$, where 
% $\mathcal{A}_{j}\define {\rm Pr}\left(u_{0}\in\mathcal{U}_{j}\right)$ ($j=1,2$) 
$\mathcal{A}_{1}$ and $\mathcal{A}_{2}$ are given in (\ref{eq:A1}) and (\ref{eq:A2}) at the top of the next page.
\end{enumerate}
Here, $\mathcal{\tilde{L}}^{(n)}_{I_{j,1C}}\left(s,r_{j,1C},r_{j,1O}\right)$ and $\mathcal{\tilde{L}}^{(n)}_{I_{j,k}}(s,r_{j,k})$ ($k\in\{1O,2\}$) are given in (\ref{eq:LTdiff_I1in}) and (\ref{eq:LTdiff_1O2}) at the top of the next page, respectively,
%\begin{align}
%&\mathcal{\tilde{L}}^{(n)}_{I_{k,j}}(s,r)=\mathcal{L}_{I_{k,j}}(s,r)\sum_{(m_{a})_{a=1}^{n}\in\mathcal{M}_{n}}\frac{n!}{\prod_{a=1}^{n}m_{a}!}\prod_{a=1}^{n}\left(\frac{2\pi}{\alpha_{J(k)}}\right)^{m_{a}}\notag\\
%&\hspace{-2mm}\times\prod_{a=1}^{n}\left(\lambda_{J(k)}s^{\frac{2}{\alpha_{J(k)}}}B^{'}\left(1+\frac{2}{\alpha_{J(k)}},a-\frac{2}{\alpha_{J(k)}},\frac{1}{1+\frac{s}{r^{\alpha_{J(k)}}}}\right)\right)^{m_{a}}
%\end{align}
where\footnote{$\mathcal{L}_{I}(s,r)$ denotes the Laplace transform of the aggregated interference $I$.} $\mathcal{L}_{I_{j,1C}}\left(s,r_{j,1C},r_{j,1O}\right)$ and $\mathcal{L}_{I_{j,k}}(s,r_{j,k})$ are given in (\ref{eq:LT_1in}) and (\ref{eq:LT_1O2}) at the top of the next page (with $J(1O)=1$ and $J(2)=2$), respectively. Moreover, $B^{'}(a,b,z)\define\int_{z}^{1}u^{a-1}(1-u)^{b-1}{\rm d}u$ ($0<z<1$), $\mathcal{N}_{n}\define\{(n_{a})_{a=1}^{3}|n_{a}\in\mathbb{N}^{0},\sum_{a=1}^{3}n_{a}=n\}$, and $\mathcal{M}_{n}\define\{(m_{a})_{a=1}^{n}|m_{a}\in\mathbb{N}^{0},\sum_{a=1}^{3}a\cdot m_{a}=n\}$.
\end{theorem}
\begin{proof}
See Appendix \ref{proof:thm1}.
\end{proof}

\emph{Theorem \ref{thm:overall_CP}} allows us to easily evaluate the coverage probability  in a numerical way. Fig.\ \ref{fig:CPvsT} plots the coverage probability versus the IN thresholds $T_{1}$ and $T_{2}$. We see from Fig.\ \ref{fig:CPvsT} that the ``Analytical" curves, which are plotted using $\mathcal{S}\left(\beta,U,T_{1},T_{2}\right)$ in \emph{Theorem \ref{thm:overall_CP}}, are reasonably close to the ``Monte Carlo" curves (the error is no larger than $5.3\%$), although \emph{Theorem \ref{thm:overall_CP}} is obtained based on some approximations (cf. Section \ref{subsec:prelim}). Compared to the non-IN scheme, we observe that the IN scheme can achieve a good coverage probability gain (up to $18.6\%$). Moreover, we see that the coverage probability depends on both $T_{1}$ and $T_{2}$. Specifically, when $T_{1}$ ($T_{2}$) is small, increasing $T_{1}$ ($T_{2}$) increases the chance of macro-BSs performing IN to near users; while when $T_{1}$ ($T_{2}$) is large, increasing $T_{1}$ ($T_{2}$) reduces the chance of macro-BSs performing IN to near users, as the resource is wasted in performing IN to users far away. 
%increasing $T_{2}$ ($T_{1}$) can increase the number of the potential IN macro-BSs of each pico-user (or macro-user). However, when $T_{2}$ ($T_{1}$) is large, each macro-BS has many IN users, which requires it to use the maximum DoF for IN; in addition, the potential IN macro-BS of each pico-user (or macro-user) may be far away, which reduces the benefits of IN. 
Further, we see that the coverage probability can be improved by jointly adjusting $T_{1}$ and $T_{2}$. 

%achieved when $T_{1}=2T_{2}$ is larger than that achieved when $T_{1}=T_{2}$. This is because (under the parameters used for plotting Fig.\ \ref{fig:CPvsT}) the macro-users need more chance for IN due to the sacrifice of DoF for IN at their serving BSs.

\begin{figure}[t]
\centering
\subfigure[$U=9$, $\beta=10$ dB]{
\includegraphics[height=0.6\columnwidth,width=0.665\columnwidth]
{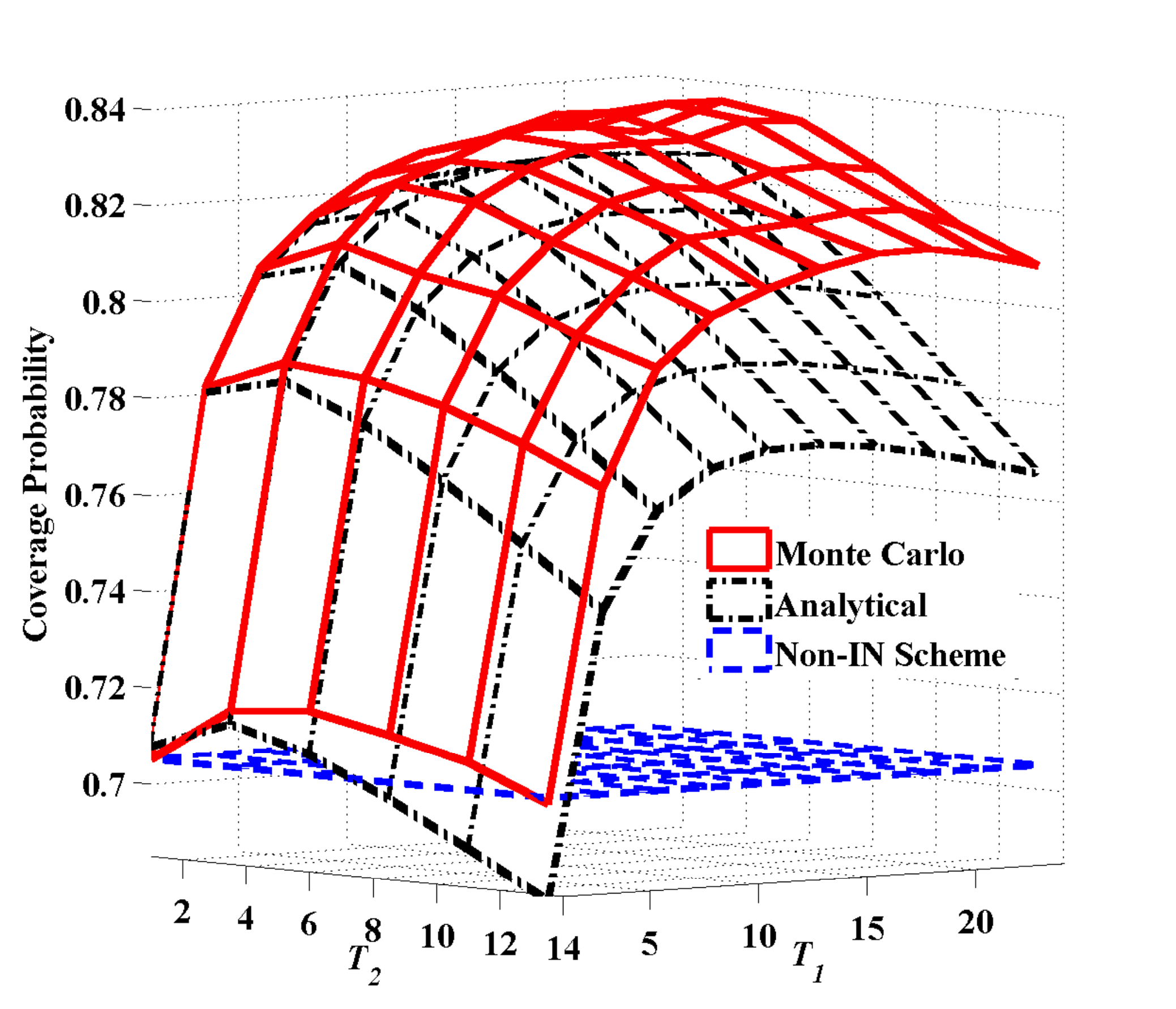}
\label{fig:CPvsT}
}
\subfigure[$\:$]{
\includegraphics[height=0.6\columnwidth,width=0.665\columnwidth]
{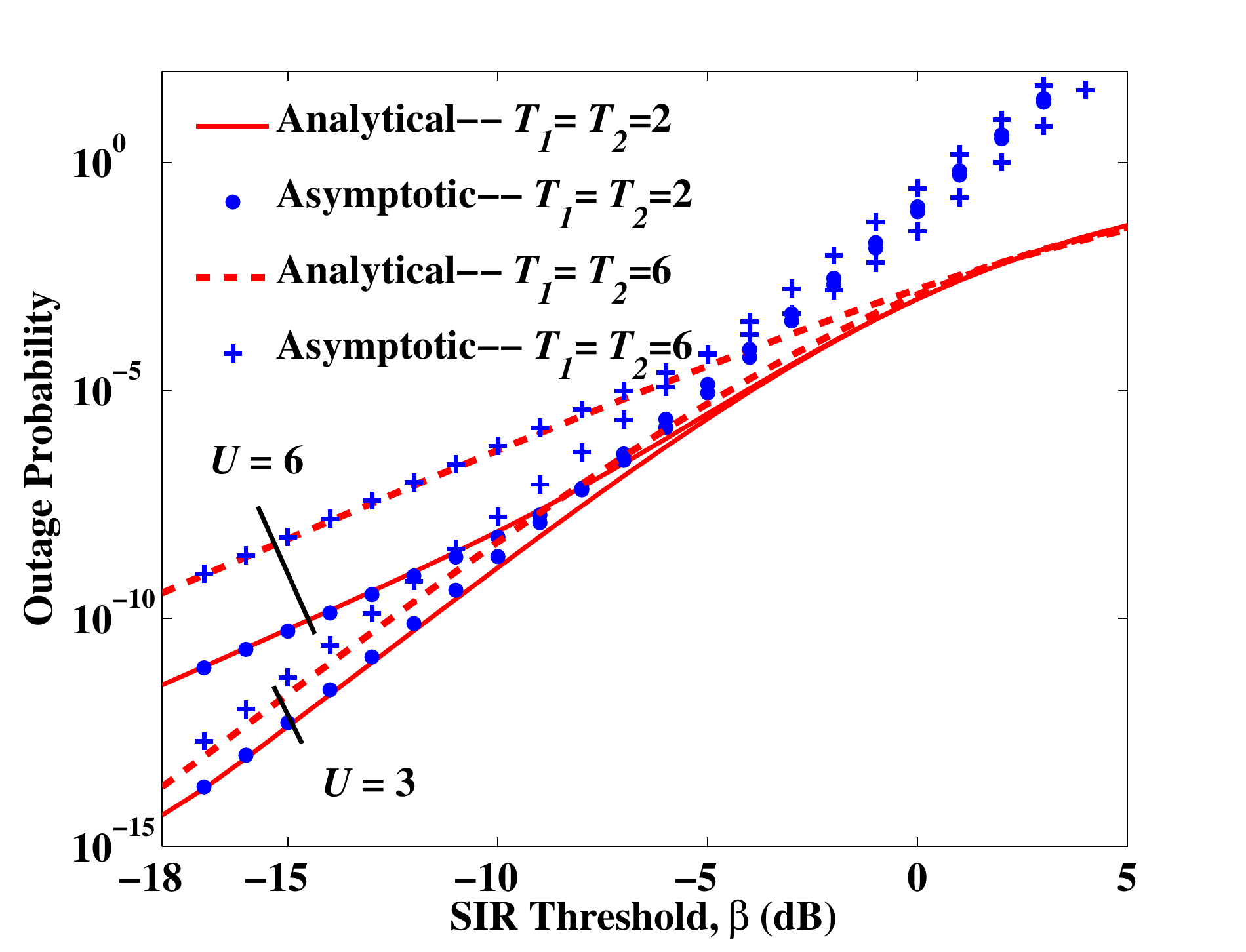}
\label{fig:OPvsbeta_smallbeta}
}
\caption{Coverage and outage Probability at $N_{1}=10$, $N_{2}=8$, $\alpha_{1}=4.5$, $\alpha_{2}=4.7$, $\frac{P_{1}}{P_{2}}=15$ dB, $\lambda_{1}=0.0005$ nodes/m$^{2}$, and $\lambda_{2}=0.001$ nodes/m$^{2}$.}\label{fig:CP_OP}
\vspace{-5mm}
\end{figure}

\subsection{Small SIR Threshold Regime}
To obtain more insights on the impact of the design parameters, in this part, we investigate the asymptotic behavior of the IN scheme in the high reliability regime where $\beta\to0$.\footnote{Note that small $\beta$ is a practical regime. For example, in current standards (e.g., 3GPP LTE), $\beta$ is very small (e.g., -7 dB) \cite{zhang14}. In addition, in wide-band systems (e.g., CDMA and UWB), $\beta$ is usually small \cite{zhang14}.}

\subsubsection{Asymptotic Outage Probability Analysis}
In this part, we analyze the asymptotic outage probability when $\beta\to0$. First, we define the \emph{order gain} of the outage probability \cite{zhang14}:
\begin{align}
d\define \lim_{\beta\to0}\frac{\log {\rm Pr}\left({\rm SIR}_{0}<\beta\right)}{\log\beta}\;.
\end{align}
Recently, a tractable approach has been proposed in \cite{haenggi14} to characterize the order gain for a class of communication schemes in wireless networks which satisfy certain conditions. However, this approach does not provide tractable analytical expressions for the \emph{coefficient} of the asymptotic outage probability, i.e., 
\begin{align}
\lim_{\beta\to0}\frac{{\rm Pr}\left({\rm SIR_{0}}<\beta\right)}{\beta^{d}}
\end{align}
%$\lim_{\beta\to0}\frac{{\rm Pr}\left({\rm SIR_{0}}<\beta\right)}{\beta^{d}}$, 
for most of the schemes using multiple antennas in the class. By utilizing series expansion of some special functions and dominated convergence theorem, we characterize both the order gain and the coefficient of the asymptotic outage probability of the IN scheme in multi-antenna HetNets, which are presented as follows.
\begin{theorem}[Asymptotic Outage Probability]\label{them:CPj_lowbeta}
Under design parameters $U$, $T_{1}$ and $T_{2}$, when $\beta\to0$, we have\footnote{$f(x)\sim g(x)$ means $\lim_{x\to0}\frac{f(x)}{g(x)}=1$.}
\begin{enumerate}
\item outage probability of a macro-user: $1-\mathcal{S}_{1}(\beta,U,T_{1},T_{2})\sim b_{1}\left(U,T_{1},T_{2}\right)\beta^{N_{1}-U}$,
%, where $U_{1}=U$ if $T_{1}>1$ and $T_{2}>1$, and $U_{1}=0$ if $T_{1}=T_{2}=1$;
\item outage probability of a pico-user: $1-\mathcal{S}_{2}(\beta,U,T_{1},T_{2})\sim b_{2}\left(U,T_{1},T_{2}\right)\beta^{N_{2}}$,
\item overall outage probability: $1-\mathcal{S}\left(\beta,U,T_{1},T_{2}\right)\sim b\left(U,T_{1},T_{2}\right)\beta^{\min\{N_{1}-U,N_{2}\}}$, where
%$b\left(U,T_{1},T_{2}\right)=$
\begin{align}
&b\left(U,T_{1},T_{2}\right)\notag\\
=&
\begin{cases}
&\hspace{-2mm}b_{2}\left(U,T_{1},T_{2}\right),\hspace{23.2mm}\;U<N_{1}-N_{2}\\
&\hspace{-2mm}b_{1}\left(U,T_{1},T_{2}\right)+b_{2}\left(U,T_{1},T_{2}\right),\;U=N_{1}-N_{2}\\
&\hspace{-2mm}b_{1}\left(U,T_{1},T_{2}\right),\hspace{23.5mm}\;U>N_{1}-N_{2}
\end{cases}
.\notag
\end{align}
\end{enumerate}
%\footnote{$1-\mathcal{S}_{j}(\beta,U,T_{1},T_{2})$ ($j=1,2$) indicates the outage probability of $u_{0}\in\mathcal{U}_{j}$, and $1-\mathcal{S}(\beta,U,T_{1},T_{2})$ indicates the overall outage probability.}
%\begin{enumerate}
%\item outage probability of a macro-user: $1-\mathcal{S}_{1}(\beta,U,T_{1},T_{2})\sim b_{1,\beta\to0}\beta^{N_{1}-U}$;
%%, where $U_{1}=U$ if $T_{1}>1$ and $T_{2}>1$, and $U_{1}=0$ if $T_{1}=T_{2}=1$;
%\item outage probability of a pico-user: $1-\mathcal{S}_{2}(\beta,U,T_{1},T_{2})\sim b_{2,\beta\to0}\beta^{N_{2}}$;
%\item overall outage probability: $1-\mathcal{S}\left(\beta,U,T_{1},T_{2}\right)\sim b_{j^{*},\beta\to0}\beta^{N_{j^{*}}-U_{j^{*}}}$, where $j^{*}=\arg\:\min_{j\in\{1,2\}}\:\left(N_{j}-U_{j}\right)$ with $U_{1}=U$ and $U_{2}=0$.
%\end{enumerate}
Here, $b_{j}\left(U,T_{1},T_{2}\right)$ is given in (\ref{eq:coeff_smallbeta}) at the top of the next page with 
\begin{align}
U_{j}=
\begin{cases}
&U,\;\;{\rm if}\; j=1, T_{1}>1, T_{2}>1\\
&0,\;\;{\rm otherwise}
\end{cases}
\;,
\end{align}
and
\begin{align}
\mathcal{P}_{j}=
\begin{cases}
&{\rm Pr}\left(u_{{\rm IN},0}=U\right),\;\;{\rm if}\; j=1, T_{1}>1, T_{2}>1\\
&1,\;\;{\rm otherwise}
\end{cases}
\;.
\end{align}
%$U_{j}=U$, $\mathcal{P}_{j}={\rm Pr}\left(u_{{\rm IN},0}=U\right)$ if $j=1$ and $T_{1}>1$, $T_{2}>1$; $U_{j}=0$, $\mathcal{P}_{j}=1$, otherwise. 
Moreover, $b_{2}\left(U,T_{1},T_{2}\right)$ decreases with $U$.
\begin{figure*}[!t]
\begin{align}\label{eq:coeff_smallbeta}
b_{j}\left(U,T_{1},T_{2}\right)&=\sum_{(n_{a})_{a=1}^{3}\in\mathcal{N}_{N_{j}-U_{j}}}\sum_{(m_{a})_{a=1}^{n_{1}}\in\mathcal{M}_{n_{1}}}\sum_{(p_{a})_{a=1}^{n_{2}}\in\mathcal{M}_{n_{2}}}\sum_{(q_{a})_{a=1}^{n_{3}}\in\mathcal{M}_{n_{3}}}\left(\int_{0}^{\infty}y^{\frac{2\alpha_{j}}{\alpha_{1}}\left(\sum_{a=1}^{n_{1}}m_{a}+\sum_{a=1}^{n_{2}}p_{a}\right)+\frac{2\alpha_{j}}{\alpha_{2}}\sum_{a=1}^{n_{2}}q_{a}}f_{Y_{j}}(y){\rm d}y\right)\notag\\
&\times \left(\prod_{a=1}^{n_{1}}\frac{1}{m_{a}!}\left(\frac{\frac{2}{\alpha_{1}}\pi \lambda_{1}}{a-\frac{2}{\alpha_{1}}}\left(\frac{P_{1}}{P_{j}}\right)^{\frac{2}{\alpha_{1}}}\right)^{m_{a}}\right)\left(\prod_{a=1}^{n_{2}}\frac{1}{p_{a}!}\left(\frac{\frac{2}{\alpha_{1}}\pi\lambda_{1}}{a-\frac{2}{\alpha_{1}}}\left(\frac{P_{1}}{P_{j}}\right)^{\frac{2}{\alpha_{1}}}\right)^{p_{a}}\right)\left(\prod_{a=1}^{n_{3}}\frac{1}{q_{a}!}\left(\frac{\frac{2}{\alpha_{2}}\pi\lambda_{2}}{a-\frac{2}{\alpha_{2}}}\left(\frac{P_{2}}{P_{j}}\right)^{\frac{2}{\alpha_{2}}}\right)^{q_{a}}\right)\notag\\
&\times \left(\prod_{a=1}^{n_{1}}\left(p_{\bar{c}}\left(U,T_{1},T_{2}\right)\left(1-\left(\frac{1}{T_{j}}\right)^{a-\frac{2}{\alpha_{1}}}\right)\right)^{m_{a}}\right)\left(\prod_{a=1}^{n_{2}}\left(\left(\frac{1}{T_{j}}\right)^{a-\frac{2}{\alpha_{1}}}\right)^{p_{a}}\right)\mathcal{P}_{j}
\end{align}
\normalsize \hrulefill
\end{figure*}
\end{theorem}
\begin{proof}
See Appendix \ref{proof:thm2}.
\end{proof}

Note that the IN scheme has three design parameters: the maximum IN DoF (i.e., $U$) and the IN thresholds (i.e., $T_{1}$ and $T_{2}$). From \emph{Theorem \ref{them:CPj_lowbeta}}, we clearly see that the maximum IN DoF and the IN thresholds affect the asymptotic behavior of the outage probability in dramatically different ways. Specifically, the maximum IN DoF $U$ can affect the order gain, while the IN thresholds can only affect the coefficient. In addition, we see that $U$ affects the order gain of the asymptotic outage probability through affecting the order gain of the asymptotic macro-user outage probability. Note that in this paper, IN is only performed at macro-BSs, and $U$ is the upper bound of the actual DoF for IN in the ZFBF precoder (which is \emph{random} due to the randomness of the network topology). Therefore, the result of the order gain in \emph{Theorem \ref{them:CPj_lowbeta}} extends the existing order gain result in \emph{single-tier} cellular networks where the DoF for IN in the ZFBF precoder is \emph{deterministic} \cite{huang13}.

Fig.\ \ref{fig:OPvsbeta_smallbeta} plots the outage probability versus SIR threshold $\beta$ for different design parameters. From Fig.\ \ref{fig:OPvsbeta_smallbeta}, we clearly see that the outage probability curves with the same maximum IN DoF $U$ have the same slope (indicating the same order gain). For the two outage probability curves with the same $U$ but different $T_{1}$ and $T_{2}$, we observe a shift between these two curves (indicating different coefficients). Therefore, Fig.\ \ref{fig:OPvsbeta_smallbeta} verifies \emph{Theorem \ref{them:CPj_lowbeta}}, and shows that the result in \emph{Theorem \ref{them:CPj_lowbeta}} can effectively reflect the outage probability for small $\beta$.

\subsubsection{Optimization of Maximum IN DoF}
From \emph{Theorem \ref{them:CPj_lowbeta}}, we know that $U$ has a larger impact on the asymptotic outage probability than the IN thresholds. In this part, we characterize the optimal maximum IN DoF which maximizes the order gain, and the optimal maximum IN DoF which minimizes the asymptotic outage probability, respectively. 

From \emph{Theorem \ref{them:CPj_lowbeta}}, we see that the order gain of the asymptotic outage probability is $\min\{N_{1}-U,N_{2}\}$. Thus, we also denote it as $d(U)$. First, we introduce the optimal design parameter $U^{*}_{d}$ which maximizes $d(U)$, i.e.,  
\begin{align}\label{eq:optU_order}
U^{*}_{d}&\define\arg\max_{U\in\{0,1,\ldots,N_{1}-1\}}d(U)\notag\\
&=\arg\max_{U\in\{0,1,\ldots,N_{1}-1\}}\min\{N_{1}-U,N_{2}\}.
\end{align}
Thus, the optimal order gain of the asymptotic outage probability is $d(U^{*}_{d})$. Next, we introduce the optimal design parameter $U^{*}(\beta,T_{1},T_{2})$ which minimizes the outage probability, i.e.,
\begin{align}\label{eq:optU_asymOP}
U^{*}(\beta,T_{1},T_{2})\define\arg\:\min_{U\in\{0,1,\ldots,N_{1}-1\}}\left(1-\mathcal{S}\left(\beta,U,T_{1},T_{2}\right)\right)\;.
\end{align}

The properties of $U_{d}^{*}$ and $U^{*}(\beta,T_{1},T_{2})$ are given by:
\begin{lemma}[Optimality Properties]\label{lem:opt_U}
$\quad$
\begin{enumerate}
\item When $\beta\to0$, we have $d\left(U^{*}_{d}\right)=N_{2}$ for all $U^{*}_{d}\in\{0,1,\ldots,N_{1}-N_{2}\}$;
\item $\exists\bar{\beta}>0$, such that for all $\beta<\bar{\beta}$, we have\footnote{As $U^{*}(\beta,T_{1},T_{2})$ is independent of $\beta$, we also write it as $U^{*}(T_{1},T_{2})$.} 
\begin{align}
&\hspace{-7.2cm}U^{*}(T_{1},T_{2})\notag\\
=
\begin{cases}
&\hspace{-2mm}N_{1}-N_{2}-1,\;{\rm if}\:b_{2}\left(N_{1}-N_{2}-1,T_{1},T_{2}\right)<\\
&\hspace{2mm}b_{1}\left(N_{1}-N_{2},T_{1},T_{2}\right)+b_{2}\left(N_{1}-N_{2},T_{1},T_{2}\right)\\
&\hspace{-2mm}N_{1}-N_{2},\;\hspace{6mm}{\rm otherwise}
\end{cases}
.
\end{align}
%$U^{*}(T_{1},T_{2})=$
%\begin{align}
%\begin{cases}
%&\hspace{-2mm}N_{1}-N_{2}-1,\;{\rm if}\:b_{2}\left(N_{1}-N_{2}-1,T_{1},T_{2}\right)<\\
%&\hspace{2mm}b_{1}\left(N_{1}-N_{2},T_{1},T_{2}\right)+b_{2}\left(N_{1}-N_{2},T_{1},T_{2}\right)\\
%&\hspace{-2mm}N_{1}-N_{2},\;\hspace{6mm}{\rm otherwise}
%\end{cases}
%.\notag
%\end{align}
\end{enumerate}
\end{lemma}
\begin{proof}
See Appendix \ref{proof:lem_optU}.
\end{proof}

%\begin{corollary}\label{corr:opt_order}
%The optimal order gain of the overall outage probability is $N_{2}$.
%\end{corollary}

Result $1)$ in \emph{Lemma \ref{lem:opt_U}} shows that $d(U_{d}^{*})$ is \emph{independent} of $U$. Thus, the IN scheme does not provide order-wise performance improvement compared to the non-IN scheme. Note that the impact of $U_{d}^{*}$ on the asymptotic outage probability is through the coefficient. Result 2) of \emph{Lemma \ref{lem:opt_U}} tells us that the asymptotic outage probability $1-\mathcal{S}\left(\beta,U^{*}(T_{1},T_{2}),T_{1},T_{2}\right)$, optimized over $U$ for given $T_{1}$ and $T_{2}$, can be further optimized by optimizing $T_{1}$ and $T_{2}$, i.e.,
\begin{align}\label{eq:opt_prob_T}
\min_{T_{1}>1,T_{2}>1} 1-\mathcal{S}\left(\beta,U^{*}(T_{1},T_{2}),T_{1},T_{2}\right)\;.
\end{align}
We shall investigate this optimization problem in future work.

%Although any $U^{*}_{d}\in\{0,1,\ldots,N_{1}-N_{2}\}$ can result in the same optimal order gain $N_{2}$, the difference lies in the coefficient $b\left(U,T_{1},T_{2}\right)$. \emph{Lemma \ref{lem:opt_U}} $2)$ shows that $U^{*}$, which minimizes the asymptotic outage probability, converges to a fixed value in the set $\{N_{1}-N_{2}-1,N_{1}-N_{2}\}$. It can be shown that the exact value of $U^{*}$ depends on whether $b_{2}\left(N_{1}-N_{2}-1,T_{1},T_{2}\right)$ is smaller than $b_{1}\left(N_{1}-N_{2},T_{1},T_{2}\right)+b_{2}\left(N_{1}-N_{2},T_{1},T_{2}\right)$ or not. As shown in \emph{Theorem \ref{them:CPj_lowbeta}}, $T_{1}$ and $T_{2}$ also have impact on the coefficient. Thus, the outage performance can be further improved by jointly optimizing $T_{1}$, $T_{2}$ and $U$, which is however beyond the scope of this paper.

\section{Conclusions}
In this paper, we propose a user-centric IN scheme in downlink two-tier multi-antenna HetNets.  Using tools from stochastic geometry, we first obtaine a tractable expression of the coverage probability. Then, we characterize the asymptotic behavior of the outage probability in the high reliability regime. The asymptotic results show that the maximum IN DoF can affect the order gain of the asymptotic outage probability, while the IN thresholds only affect the coefficient of the asymptotic outage probability. Moreover, we show that the IN scheme can linearly improve the outage performance, and characterize the optimal maximum IN DoF which minimizes the asymptotic outage probability. The analytical results obtained in this paper provide valuable design insights for practical HetNets.

\appendix

\subsection{Proof of Lemma \ref{lem:num_req_pmf}}\label{proof:pmf_Lbar}
According to Slivnyak's theorem \cite{haenggi09}, we focus on a macro-BS located at origin. We refer to this macro-BS as macro-BS $0$. Note that both scheduled macro-users and scheduled pico-users may send IN requests to macro-BS $0$. We first characterize the probability that a scheduled macro-user sends an IN request to macro-BS $0$. Denote $R_{1i}$ as the distance between macro-BS $0$ and a randomly selected (according to the uniform distribution) scheduled macro-user which we refer to as scheduled macro-user $i$. According to Section \ref{subsec:IN_descp}, scheduled macro-user $i$ sends an IN request to macro-BS $0$ if
\begin{align}
\frac{P_{1}Y_{1}^{-\alpha_{1}}}{P_{1}R_{1i}^{-\alpha_{1}}}<T_{1}\;\;{\rm and}\;\;Y_{1}<R_{1}\;.
\end{align}
Hence, scheduled macro-user $i$ sends an IN request to macro-BS $0$ with probability $p_{1i}(T_{1})$ which is given by
\begin{align}
p_{1i}(T_{1})={\rm Pr}\left(T_{1}^{-\frac{1}{\alpha_{1}}}R_{1i}<Y_{1}<R_{1i}\right)\;.
\end{align}
Assume that the scheduled macro-users form a homogeneous PPP with density $\lambda_{1}$. Conditioned on $R_{1i}=r$, we have
%\footnote{Here, we make explicit the dependence of $p_{1i}$ on $x$.}
\begin{align}\label{eq:prob_macro_send}
p_{1i,R_{1i}}(r,T_{1})&={\rm Pr}\left(T_{1}^{-\frac{1}{\alpha_{1}}}r<Y_{1}<r\right)\notag\\
&=\int_{T_{1}^{-\frac{1}{\alpha_{1}}}r}^{r}f_{Y_{1}}(y){\rm d}y
\end{align}
where $f_{Y_{1}}(y)$ is the probability density function (p.d.f.) of $Y_{1}$ given by (\ref{eq:pdfY1}). Then, the scheduled macro-user density at distance $r$ away from macro-BS $0$ is $p_{1i,R_{1i}}(r,T_{1})\lambda_{1}$. This indicates that the scheduled macro-users which send IN requests to macro-BS $0$ form an inhomogeneous PPP with density $p_{1i,R_{1i}}(r,T_{1})\lambda_{1}$ at distance $r$ away from macro-BS $0$.
%Thus, the average number of IN requests sent by scheduled macro-users to macro-BS $0$ is
%\begin{align}\label{eq:req_macro}
%\bar{L}_{1}&=2\pi\int_{0}^{\infty}r p_{1i,R_{1i}}(r,T_{1})\lambda_{1}{\rm d}r\notag\\
%&=2\pi\lambda_{1}\int_{0}^{\infty}r\int_{T_{1}^{-\frac{1}{\alpha_{1}}}r}^{r}f_{Y_{1}}(y){\rm d}y{\rm d}r\;.
%\end{align}

Next, we characterize the probability that a scheduled pico-user sends an IN request to macro-BS $0$. Denote $R_{2i}$ as the distance between macro-BS $0$ and a randomly selected (according to the uniform distribution) scheduled pico-user which refer to as scheduled pico-user $i$. Similarly, we assume that the scheduled pico-users form a homogeneous PPP with density $\lambda_{2}$, and it is independent of the PPP formed by the scheduled macro-users. Then, we can show that the scheduled pico-users which send IN requests to macro-BS $0$ form an inhomogeneous PPP with density $p_{2i,R_{2i}}(r,T_{2})\lambda_{2}$ at distance $r$ away from macro-BS $0$, where
\begin{align}\label{eq:prob_pico_send}
p_{2i,R_{2i}}(r,T_{2})&={\rm Pr}\left(\left(\frac{P_{2}}{P_{1}T_{2}}\right)^{\frac{1}{\alpha_{2}}}r^{\frac{\alpha_{1}}{\alpha_{2}}}<Y_{2}<\left(\frac{P_{2}}{P_{1}}\right)^{\frac{1}{\alpha_{2}}}r^{\frac{\alpha_{1}}{\alpha_{2}}}\right)\notag\\
&=\int_{\left(\frac{P_{2}}{P_{1}T_{2}}\right)^{\frac{1}{\alpha_{2}}}r^{\frac{\alpha_{1}}{\alpha_{2}}}}^{\left(\frac{P_{2}}{P_{1}}\right)^{\frac{1}{\alpha_{2}}}r^{\frac{\alpha_{1}}{\alpha_{2}}}}f_{Y_{2}}(y){\rm d}y
\end{align}
where $f_{Y_{2}}(y)$ is the p.d.f. of $Y_{2}$ given by (\ref{eq:pdfY2}).
%Thus, the average number of IN requests sent by scheduled pico-users to macro-BS $0$ is
%\begin{align}\label{eq:req_pico}
%\bar{L}_{2}&=2\pi\int_{0}^{\infty}r p_{2i,R_{2i}}(r,T_{2})\lambda_{2}{\rm d}r\notag\\
%&=2\pi\lambda_{2}\int_{0}^{\infty}r\int_{\left(\frac{P_{2}}{P_{1}T_{2}}\right)^{\frac{1}{\alpha_{2}}}r^{\frac{\alpha_{1}}{\alpha_{2}}}}^{\left(\frac{P_{2}}{P_{1}}\right)^{\frac{1}{\alpha_{2}}}r^{\frac{\alpha_{1}}{\alpha_{2}}}}f_{Y_{2}}(y){\rm d}y{\rm d}r\;.
%\end{align}

According to the superposition property of PPPs \cite{haenggi09}, the scheduled macro-users and the scheduled pico-users which send IN requests to macro-BS $0$, i.e., the potential IN users of macro-BS $0$, still form a PPP with density $p_{1i,R_{1i}}(r,T_{1})\lambda_{1}+p_{2i,R_{2i}}(r,T_{2})\lambda_{2}$ at distance $r$ away from macro-BS $0$. Therefore, the number of the potential IN users of macro-BS $0$ is Poisson distributed with the following parameter (i.e., mean)
\begin{align}\label{eq:req_pico_macro}
\bar{L}&=2\pi\int_{0}^{\infty}r \left(p_{1i,R_{1i}}(r,T_{1})\lambda_{1}+p_{2i,R_{2i}}(r,T_{2})\lambda_{2}\right){\rm d}r\notag\\
&=\bar{L}_{1}+\bar{L}_{2}\;.
%&=2\pi\lambda_{2}\int_{0}^{\infty}r\int_{\left(\frac{P_{2}}{P_{1}T_{2}}\right)^{\frac{1}{\alpha_{2}}}r^{\frac{\alpha_{1}}{\alpha_{2}}}}^{\left(\frac{P_{2}}{P_{1}}\right)^{\frac{1}{\alpha_{2}}}r^{\frac{\alpha_{1}}{\alpha_{2}}}}f_{Y_{2}}(y){\rm d}y{\rm d}r\;.
\end{align}

%Moreover, note that we assume the scheduled macro-users and the scheduled pico-users form two independent homogeneous PPPs. Then, we have
%\begin{align}
%\bar{L}=\bar{L}_{1}+\bar{L}_{2}\;,
%\end{align}
%which completes the proof.

\subsection{Proof of Lemma \ref{lem:INprob}}\label{proof:lem_INprob}
From Section \ref{subsec:IN_descp}, we know that whether a scheduled user sends an IN request to a macro-BS or not depends on its location relative to this macro-BS. Hence, the event that $u_{0}$ sends an IN request to one of its potential IN macro-BSs and the event that any other scheduled user sends an IN request to the same macro-BS are dependent. The dependence is especially high when $u_{0}$ and that scheduled user are close. For analytical tractability, we approximate these two events as independent events. Then, we have
\begin{align}\label{eq:INprob_cal}
&p_{c}\left(U,T_{1},T_{2}\right)\approx {\rm E}\left[\min\left\{\frac{U}{K_{0}+1},1\right\}\right]\notag\\
=&\sum_{k=0}^{U-1}{\rm Pr}\left(K_{0}=k\right)+\sum_{k=U}^{\infty}\frac{U}{k+1}{\rm Pr}\left(K_{0}=k\right)\;.
\end{align}
Substituting (\ref{eq:K_pmf}) into (\ref{eq:INprob_cal}), we have the final result.

\subsection{Proof of Theorem \ref{thm:overall_CP}}\label{proof:thm1}
Let $R_{j,1C}$ and $R_{j,1O}$ denote the minimum and maximum possible distances between $u_{0}\in\mathcal{U}_{j}$ and its nearest and furthest macro-interferers (among $u_{0}$'s potential IN macro-BSs which do not select $u_{0}$ for IN), respectively. Let $R_{j,2}$ denote the minimum possible distance between $u_{0}\in\mathcal{U}_{2}$ and its nearest pico-interferer. The relationships between $R_{j,1C}$, $R_{j,1O}$, $R_{j,2}$, and $Y_{j}$, respectively, are shown in Table \ref{tab:para_B2larger}. Based on (\ref{eq:SIRj0}) and conditioned on $Y_{j}=y$, we have
\begin{align}
&{\rm Pr}\left({\rm SIR}_{0}>\beta|u_{0}\in\mathcal{U}_{j},Y_{j}=y\right)\notag\\
=&{\rm Pr}\left(\left|\mathbf{h}_{j,00}^{\dagger}\mathbf{f}_{j,0}\right|^{2}>\beta y^{\alpha_{j}}\left(\frac{P_{1}}{P_{j}}I_{j,1C}+\frac{P_{1}}{P_{j}}I_{j,1O}+\frac{P_{2}}{P_{j}}I_{j,2}\right)\right)\notag\\
=&\sum_{n=0}^{M_{j}-1}\frac{\left(\beta y^{\alpha_{j}}\right)^{n}}{n!}\sum_{\left(n_{a}\right)_{a=1}^{3}\in\mathcal{N}_{n}}\binom{n}{n_{1},n_{2},n_{3}}\left(\frac{P_{1}}{P_{j}}\right)^{n_{1}+n_{2}}\left(\frac{P_{2}}{P_{j}}\right)^{n_{3}}\notag\\
&\hspace{0mm}\times {\rm E}_{I_{j,1C}}\left[I_{j,1C}^{n_{1}}\exp\left(-\beta y^{\alpha_{j}}\frac{P_{1}}{P_{j}}I_{j,1C}\right)\right]\notag\\
&\hspace{0mm}\times{\rm E}_{I_{j,1O}}\left[I_{j,1O}^{n_{2}}\exp\left(-\beta y^{\alpha_{j}}\frac{P_{1}}{P_{j}}I_{j,1O}\right)\right]\notag\\
&\hspace{0mm}\times {\rm E}_{I_{j,2}}\left[I_{j,2}^{n_{3}}\exp\left(-\beta y^{\alpha_{j}}\frac{P_{2}}{P_{j}}I_{j,2}\right)\right]\notag\\
=&\sum_{n=0}^{M_{j}-1}\frac{\left(-\beta y^{\alpha_{j}}\right)^{n}}{n!}\sum_{\left(n_{a}\right)_{a=1}^{3}\in\mathcal{N}_{n}}\binom{n}{n_{1},n_{2},n_{3}}\left(\frac{P_{1}}{P_{j}}\right)^{n_{1}+n_{2}}\notag\\
&\times \left(\frac{P_{2}}{P_{j}}\right)^{n_{3}}\mathcal{L}^{(n_{1})}_{I_{j,1C}}\left(s,r_{j,1C},r_{j,1O}\right)|_{s=\beta y^{\alpha_{j}}\frac{P_{1}}{P_{j}}}\notag\\
&\times \mathcal{L}^{(n_{2})}_{I_{j,1O}}\left(s,r_{j,1O}\right)|_{s=\beta y^{\alpha_{j}}\frac{P_{1}}{P_{j}}}\mathcal{L}_{I_{j,2}}^{(n_{3})}(s,r_{j,2})|_{s=\beta y^{\alpha_{j}}\frac{P_{2}}{P_{j}}}
\end{align}
where $\mathcal{L}^{(n)}_{I}(s,r)$ denotes the $n$th-order derivative of the Laplace transform $\mathcal{L}_{I}(s,r)$.

\begin{table}[t]
\caption{Parameter values}\label{tab:para_B2larger}
\begin{center}
\vspace{-6mm}
\begin{tabular}{|c|c|c|c|c|c|}
\hline
$j$&$r_{j,1C}$&$r_{j,1O}$&$r_{j,2}$\\
\hline
$1$ &$Y_{1}$&$T_{1}^{\frac{1}{\alpha_{1}}}Y_{1}$ &$\left(\frac{P_{2}}{P_{1}}\right)^{\frac{1}{\alpha_{2}}}Y_{1}^{\frac{\alpha_{1}}{\alpha_{2}}}$\\
\hline
$2$&$\left(\frac{P_{1}}{P_{2}}\right)^{\frac{1}{\alpha_{1}}}Y_{2}^{\frac{\alpha_{2}}{\alpha_{1}}}$ & $\left(\frac{P_{1}}{P_{2}}T_{2}\right)^{\frac{1}{\alpha_{1}}}Y_{2}^{\frac{\alpha_{2}}{\alpha_{1}}}$&$Y_{2}$\\
\hline
\end{tabular}
\end{center}
\vspace{-6mm}
\end{table}

Now, we calculate $\mathcal{L}_{I}(s,r)$ and $\mathcal{L}^{(n)}_{I}(s,r)$, respectively.  First, $\mathcal{L}_{I_{j,1C}}(s,r_{j,1C},r_{j,1O})$ can be calculated as follows:
\begin{align}\label{eq:LT_IC}
&\mathcal{L}_{I_{j,1C}}(s,r_{j,1C},r_{j,1O})\notag\\
=&{\rm E}_{\Phi_{1C},\{\mathbf{g}_{1,\ell}\}}\left[\exp(-s\sum_{\ell\in\Phi_{1C}}\left|D_{1,\ell0}\right|^{-\alpha_{1}}\mathbf{g}_{1,\ell}\right]\notag\\
\eqla&{\rm E}_{\Phi_{1C}}\left[\prod_{\ell\in\Phi_{1C}}{\rm E}_{\{\mathbf{g}_{1,\ell}\}}\left[\exp(-s\left|D_{1,\ell0}\right|^{-\alpha_{1}}\mathbf{g}_{1,\ell}\right]\right]\notag\\
=&{\rm E}_{\Phi_{1C}}\left[\prod_{\ell\in\Phi_{1C}}\frac{1}{1+s|D_{1,\ell0}|^{-\alpha_{1}}}\right]\notag\\
\eqlb&\exp\left(-2\pi p_{\bar{c}}\left(U,T_{1},T_{2}\right)\lambda_{1}\int_{r_{j,1C}}^{r_{j,1O}}\left(1-\frac{1}{1+\frac{s}{r^{\alpha_{1}}}}\right)r{\rm d}r\right)
%\eqlc&\exp\left(-\frac{2\pi p_{\bar{c}}\left(U,T_{1},T_{2}\right)\lambda_{1}s^{\frac{2}{\alpha_{1}}}}{\alpha_{1}}\left(\right)\right)
\end{align}
where $\mathbf{g}_{1,\ell}=\left|\mathbf{h}_{1,\ell0}^{\dagger}\mathbf{f}_{1,\ell}\right|^{2}$, (a) is obtained by noting that $\mathbf{g}_{1,\ell}$ ($\ell\in\Phi_{1C}$) are mutually independent, and (b) is obtained by using the probability generating functional of a PPP \cite{haenggi09}. Further, by first letting $s^{-\frac{1}{\alpha_{1}}}r=t$ and then $\frac{1}{1+t^{-\alpha_{1}}}=w$, we obtain the result in (\ref{eq:LT_1in}).

Next, based on (\ref{eq:LT_IC}) and utilizing Fa${\rm \grave{a}}$ di Bruno's formula \cite{johnson02}, $\mathcal{L}_{I_{j,1C}}^{(n_{1})}(s,r_{j,1C},r_{j,1O})$ can be calculated as follows:
\begin{align}
&\mathcal{L}_{I_{j,1C}}^{(n_{1})}(s,r_{j,1C},r_{j,1O})\notag\\
=&\mathcal{L}_{I_{j,1C}}(s,r_{j,1C},r_{j,1O})\sum_{(m_{a})_{a=1}^{n}\in\mathcal{M}_{n_{1}}}\frac{n_{1}!}{\prod_{a=1}^{n_{1}}(a!)^{m_{a}}m_{a}!}\notag\\
&\hspace{-5mm}\times\prod_{a=1}^{n_{1}}\left(2\pi p_{\bar{c}}\left(U,T_{1},T_{2}\right)\lambda_{1}\int_{r_{j,1C}}^{r_{j,1O}}\frac{{\rm d}^{a}}{{\rm d}s^{a}}\left(\frac{1}{1+\frac{s}{r^{\alpha_{1}}}}\right)r{\rm d}r\right)^{m_{a}}\notag\\
=&\mathcal{L}_{I_{j,1C}}(s,r_{j,1C},r_{j,1O})\sum_{(m_{a})_{a=1}^{n}\in\mathcal{M}_{n_{1}}}\frac{n_{1}!(-1)^{n_{1}}}{\prod_{a=1}^{n_{1}}(a!)^{m_{a}}m_{a}!}\notag\\
&\hspace{-5mm}\times\prod_{a=1}^{n_{1}}\left(2\pi p_{\bar{c}}\left(U,T_{1},T_{2}\right)\lambda_{1}\Gamma\left(a+1\right)\int_{r_{j,1C}}^{r_{j,1O}}\frac{r^{1-a\alpha_{1}}}{\left(1+\frac{s}{r^{\alpha_{1}}}\right)^{a+1}}{\rm d}r\right)^{m_{a}}
\end{align}
where the integral can be solved using similar method as calculating $\mathcal{L}_{I_{j,1C}}(s,r_{j,1C},r_{j,1O})$. Similarly, we can calculate $\mathcal{L}_{I_{j,1O}}\left(s,r_{j,1O}\right)$, $\mathcal{L}^{(n_{2})}_{I_{j,1O}}\left(s,r_{j,1O}\right)$, $\mathcal{L}_{I_{j,2}}(s,r_{j,2})$ and $\mathcal{L}_{I_{j,2}}^{(n_{3})}(s,r_{j,2})$. Finally, removing the conditions on $Y_{j}=y$ and after some algebraic manipulations, we can obtain the final result.
%where $g^{(a)}(s,r_{1,ij},r_{1,oj})$ is the $a$th-order derivative of $g(s,r_{1,ij},r_{1,oj})$, and $g(s,r_{1,ij},r_{1,oj})=-2\pi p_{\bar{c}}\left(U,T_{1},T_{2}\right)\lambda_{1}\int_{r_{1,ij}}^{r_{1,oj}}\left(1-\frac{1}{1+\frac{s}{r^{\alpha_{1}}}}\right)r{\rm d}r$. The integral in

\subsection{Proof of Theorem \ref{them:CPj_lowbeta}}\label{proof:thm2}
Conditioned on $Y_{j}=y$, we have
\begin{align}\label{eq:OPj_condi}
&1-{\rm Pr}\left({\rm SIR}_{0}>\beta|u_{0}\in\mathcal{U}_{j},Y_{j}=y\right)\notag\\
=&{\rm Pr}\left(\left|\mathbf{h}_{j,00}^{\dagger}\mathbf{f}_{j,0}\right|^{2}\le\beta y^{\alpha_{j}}\left(\frac{P_{1}}{P_{j}}I_{j,1C}+\frac{P_{1}}{P_{j}}I_{j,1O}+\frac{P_{2}}{P_{j}}I_{j,2}\right)\right)\notag\\
=&\exp\left(-\beta y^{\alpha_{j}}\left(\frac{P_{1}}{P_{j}}I_{j,1C}+\frac{P_{1}}{P_{j}}I_{j,1O}+\frac{P_{2}}{P_{j}}I_{j,2}\right)\right)\notag\\
&\times \sum_{n=M_{j}}^{\infty}\frac{\left(\beta y^{\alpha_{j}}\right)^{n}}{n!}\left(\frac{P_{1}}{P_{j}}I_{j,1C}+\frac{P_{1}}{P_{j}}I_{j,1O}+\frac{P_{2}}{P_{j}}I_{j,2}\right)^{n}\notag\\
=&\sum_{n=M_{j}}^{\infty}\frac{\left(-\beta y^{\alpha_{j}}\right)^{n}}{n!}\sum_{\left(n_{a}\right)_{a=1}^{3}\in\mathcal{N}_{n}}\binom{n}{n_{1},n_{2},n_{3}}\left(\frac{P_{1}}{P_{j}}\right)^{n_{1}+n_{2}}\notag\\
&\hspace{-3mm}\times \mathcal{L}^{(n_{1})}_{I_{j,1C}}\left(s,r_{j,1C},r_{j,1O}\right)|_{s=\beta y^{\alpha_{j}}\frac{P_{1}}{P_{j}}}\mathcal{L}^{(n_{2})}_{I_{j,1O}}\left(s,r_{j,1O}\right)|_{s=\beta y^{\alpha_{j}}\frac{P_{1}}{P_{j}}}\notag\\
&\hspace{-3mm}\times \left(\frac{P_{2}}{P_{j}}\right)^{n_{3}}\mathcal{L}_{I_{j,2}}^{(n_{3})}(s,r_{j,2})|_{s=\beta y^{\alpha_{j}}\frac{P_{2}}{P_{j}}}\;.
\end{align}
%where $\mathcal{L}^{(n_{1})}_{I_{j,1C}}\left(s,r_{1,ij},r_{1,oj}\right)$, $\mathcal{L}^{(n_{2})}_{I_{j,1O}}\left(s,r_{1,oj}\right)$, and $\mathcal{L}_{I_{j,2}}^{(n_{3})}(s,r_{2,j})$ are given in (\ref{eq:LTdiff_I1in}) and (\ref{eq:LTdiff_1O2}).
Similar to the calculations in Appendix \ref{proof:thm1}, after some algebraic manipulations and removing the condition on $Y_{j}=y$, we have
\begin{align}\label{eq:OP_j_sumtoinf}
&1-\mathcal{S}_{j}\left(U,T_{1},T_{2},\beta\right)\notag\\
=&\int_{0}^{\infty}\sum_{n=M_{j}}^{\infty}\mathcal{T}_{j,Y_{j}}\left(n,y,U,T_{1},T_{2},\beta\right)f_{Y_{j}}(y){\rm d}y
\end{align}
where
\begin{align}
\mathcal{T}_{j,Y_{j}}\left(n,y,U,T_{1},T_{2},\beta\right)=&\frac{1}{n!}\sum_{(n_{a})_{a=1}^{3}\in\mathcal{N}_{n}}\binom{n}{n_{1},n_{2},n_{3}}\notag\\
&\hspace{-3cm}\times\mathcal{\tilde{L}}^{(n_{1})}_{I_{j,1C}}\left(s,y\right)|_{s=\beta y^{\alpha_{j}}\frac{P_{1}}{P_{j}}}\mathcal{\tilde{L}}^{(n_{2})}_{I_{j,1O}}\left(s,y\right)|_{s=\beta y^{\alpha_{h}}\frac{P_{1}}{P_{j}}}\notag\\
&\hspace{-3cm}\times\mathcal{\tilde{L}}^{(n_{3})}_{I_{j,2}}\left(s,y\right)|_{s=\beta y^{\alpha_{j}}\frac{P_{2}}{P_{j}}}\;.
\end{align}

Now, we calculate the asymptotic outage probability when $\beta\to0$, i.e.,
\begin{align}
\lim_{\beta\to0}\int_{0}^{\infty}\sum_{n=M_{j}}^{\infty}\mathcal{T}_{j,Y_{j}}\left(n,y,U,T_{1},T_{2},\beta\right)f_{Y_{j}}(y){\rm d}y\;.
\end{align}
We note that
\begin{align}
B^{'}(a,b,z)=\frac{(1-z)^{b}}{b}+o\left((1-z)^{b}\right)\;,\quad{\rm as}\; z\to1\;.
\end{align}
Then, we have
{\small\begin{align}
B^{'}\left(\frac{2}{\alpha},1-\frac{2}{\alpha},\frac{1}{1+c\beta}\right)&= \frac{\left(c\beta\right)^{1-\frac{2}{\alpha}}}{1-\frac{2}{\alpha}}+o\left(\beta^{1-\frac{2}{\alpha}}\right)\;,\\
B^{'}\left(1+\frac{2}{\alpha},a-\frac{2}{\alpha},\frac{1}{1+c\beta}\right)&= \frac{(c\beta)^{a-\frac{2}{\alpha}}}{a-\frac{2}{\alpha}}+o\left(\beta^{a-\frac{2}{\alpha}}\right)\;,
\end{align}}where $c\in\mathbb{R}^{+}$. Based on these two asymptotic expressions,
%and let $s=\tilde{c}\beta$ ($\tilde{c}\in\mathbb{R}^{+}$) in $\mathcal{L}^{(n)}_{I}\left(s,r\right)$,
we can obtain\footnote{$f(x)=o\left(g(x)\right)$ means $\lim_{x\to0}\frac{f(x)}{g(x)}=0$.}
\begin{align}
&\mathcal{\tilde{L}}^{(n_{1})}_{I_{j,1C}}\left(s,y\right)=\beta^{n_{1}}\sum_{(m_{a})_{a=1}^{n_{1}}\in\mathcal{M}_{n_{1}}}\frac{n_{1}!}{\prod_{a=1}^{n_{1}}m_{a}!}\notag\\
&\hspace{-3mm}\times\prod_{a=1}^{n_{1}}\left(\frac{\frac{2\pi}{\alpha_{1}}p_{\bar{c}}\left(U,T_{1},T_{2}\right)\lambda_{1}}{a-\frac{2}{\alpha_{1}}}\left(\frac{P_{1}y^{\alpha_{j}}}{P_{j}}\right)^{\frac{2}{\alpha_{1}}}\left(1-\left(\frac{1}{T_{j}}\right)^{a-\frac{2}{\alpha_{1}}}\right)\right)^{m_{a}}\notag\\
&\hspace{2mm}+o\left(\beta^{n_{1}}\right)\;,\label{eq:LT_1C_lowbeta}\\
&\mathcal{\tilde{L}}^{(n_{2})}_{I_{j,1O}}\left(s,y\right)=\beta^{n_{2}}\sum_{(p_{a})_{a=1}^{n_{2}}\in\mathcal{M}_{n_{2}}}\frac{n_{2}!}{\prod_{a=1}^{n_{2}}p_{a}!}\notag\\
&\hspace{2mm}\times \prod_{a=1}^{n_{2}}\left(\frac{\frac{2\pi}{\alpha_{1}}\lambda_{1}}{a-\frac{2}{\alpha_{1}}}\left(\frac{P_{1}}{P_{j}}\right)^{\frac{2}{\alpha_{1}}}y^{\frac{2\alpha_{j}}{\alpha_{1}}}\left(\frac{1}{T_{j}}\right)^{a-\frac{2}{\alpha_{1}}}\right)^{p_{a}}+o\left(\beta^{n_{2}}\right)\;,\label{eq:LT_1O_lowbeta}\\
&\mathcal{\tilde{L}}^{(n_{3})}_{I_{j,2}}\left(s,y\right)=\beta^{n_{3}}\sum_{(q_{a})_{a=1}^{n_{3}}\in\mathcal{M}_{n_{3}}}\frac{n_{3}!}{\prod_{a=1}^{n_{2}}q_{a}!}\notag\\
&\hspace{2mm}\times \prod_{a=1}^{n_{3}}\left(\frac{\frac{2\pi}{\alpha_{2}}\lambda_{2}}{a-\frac{2}{\alpha_{2}}}\left(\frac{P_{2}}{P_{j}}\right)^{\frac{2}{\alpha_{2}}}y^{\frac{2\alpha_{j}}{\alpha_{2}}}\right)^{q_{a}}+o\left(\beta^{n_{3}}\right)\;.\label{eq:LT_2_lowbeta}
\end{align}
Moreover, utilizing dominated convergence theorem, we can show that
\begin{align}\label{eq:OP_lowbeta_change}
&\lim_{\beta\to0}\int_{0}^{\infty}\sum_{n=M_{j}}^{\infty}\mathcal{T}_{j,Y_{j}}\left(n,y,U,T_{1},T_{2},\beta\right)f_{Y_{j}}(y){\rm d}y\notag\\
=&\int_{0}^{\infty}\sum_{n=M_{j}}^{\infty}\lim_{\beta\to0}\mathcal{T}_{j,Y_{j}}\left(n,y,U,T_{1},T_{2},\beta\right)f_{Y_{j}}(y){\rm d}y\;.
\end{align}
Hence, substituting (\ref{eq:LT_1C_lowbeta}), (\ref{eq:LT_1O_lowbeta}) and (\ref{eq:LT_2_lowbeta}) into (\ref{eq:OP_j_sumtoinf}), and after some algebraic manipulations, we obtain Results 1), 2) and 3) in \emph{Theorem \ref{them:CPj_lowbeta}}. To complete the proof, we now show that $b_{2}\left(U,T_{1},T_{2}\right)$ decreases with $U$. This can be proved by noting that i) $b_{2}\left(U,T_{1},T_{2}\right)$ is an increasing function of $p_{\bar{c}}\left(U,T_{1},T_{2}\right)$, and ii) $p_{\bar{c}}\left(U,T_{1},T_{2}\right)$ decreases with $U$ (which can be easily shown using (\ref{eq:INprob_cal})).

\subsection{Proof of Lemma \ref{lem:opt_U}}\label{proof:lem_optU}
\subsubsection{Proof of Result 1)}
It can be easily shown that when $U\in\{0,1,\ldots,N_{1}-N_{2}\}$, we have $N_{1}-U\ge N_{2}$. Thus, $d(U)=N_{2}$. When $U\in\{ N_{1}-N_{2}+1,\ldots,N_{1}-1\}$, we have $N_{1}-U< N_{2}$. Thus, $d(U)=N_{1}-U<N_{2}$. Hence,  by (\ref{eq:optU_order}), we have $U_{d}^{*}=N_{2}$.

\subsubsection{Proof of Result 2)}
First, we note that $U^{*}(\beta,T_{1},T_{2})$ can achieve the optimal order gain, i.e., $U^{*}(\beta,T_{1},T_{2})\in\{0,1,\ldots,N_{1}-N_{2}\}$. Next, we compare the coefficients of $\beta^{N_{2}}$ which are achieved by different $U\in\{0,1,\ldots,N_{1}-N_{2}\}$. We consider two cases. i) When $U<N_{1}-N_{2}$, as $b_{2}\left(U,T_{1},T_{2}\right)$ decreases with $U$, the coefficients are $b_{2}\left(N_{1}-N_{2}-1,T_{1},T_{2}\right)<b_{2}\left(N_{1}-N_{2}-2,T_{1},T_{2}\right)<\ldots<b_{2}\left(0,T_{1},T_{2}\right)$. ii) When $U=N_{1}-N_{2}$, the coefficient of $\beta^{N_{2}}$ is $b_{1}\left(N_{1}-N_{2},T_{1},T_{2}\right)+b_{2}\left(N_{1}-N_{2},T_{1},T_{2}\right)$.

Based on the discussions above and (\ref{eq:optU_asymOP}), we know that the exact value of $U^{*}(\beta,T_{1},T_{2})$ depends on whether $b_{2}\left(N_{1}-N_{2}-1,T_{1},T_{2}\right)$ is smaller than $b_{1}\left(N_{1}-N_{2},T_{1},T_{2}\right)+b_{2}\left(N_{1}-N_{2},T_{1},T_{2}\right)$ or not. In particular, when $b_{2}\left(N_{1}-N_{2}-1,T_{1},T_{2}\right)<b_{1}\left(N_{1}-N_{2},T_{1},T_{2}\right)+b_{2}\left(N_{1}-N_{2},T_{1},T_{2}\right)$, we have $U^{*}(\beta,T_{1},T_{2})=N_{1}-N_{2}-1$; otherwise, $U^{*}(\beta,T_{1},T_{2})=N_{1}-N_{2}$.

%\bibliographystyle{IEEEtran}
%\bibliography{IEEEabrv,journals,books,confs}


\end{document}